\documentclass[12pt,english]{article}
\usepackage[english]{babel}
\usepackage{feynmp-auto}
\usepackage{amsmath,amssymb,amscd,amsthm,xcolor}
\usepackage{amsfonts}
\usepackage{mathrsfs}
\usepackage{mathtools}
\usepackage{float}
\usepackage{graphicx}
\usepackage{url}
\usepackage{hyperref}
\usepackage{cite}
\usepackage{physics}
\usepackage{braket}
\usepackage{verbatim}
\usepackage[font=footnotesize,labelfont=bf,width=.9\textwidth]{caption}
\usepackage{multirow}
\usepackage{boldline}
\usepackage{enumitem}
\usepackage{wrapfig}

\newtheorem{theorem}{Theorem}

\usepackage{tikz}
\usetikzlibrary{knots,decorations.markings,decorations.pathmorphing,decorations.pathreplacing, shapes.misc,perspective,arrows.meta}
\usetikzlibrary{arrows}
\usetikzlibrary{external}

\usepackage{scalerel}
\usepackage[normalem]{ulem}
\allowdisplaybreaks

\hypersetup{
    colorlinks,%
    citecolor=blue,%
    filecolor=blue,%
    linkcolor=blue,%
    urlcolor=blue
}

\makeatletter
\renewcommand\section{\@startsection {section}{1}{\z@}%
                                   {-5.5ex \@plus -1ex \@minus -.2ex}
                                   {2.3ex \@plus.2ex}%
                                   {\normalfont\large\bfseries}}
\renewcommand\subsection{\@startsection{subsection}{2}{\z@}%
                                     {-3.25ex\@plus -1ex \@minus -.2ex}%
                                     {1.5ex \@plus .2ex}%
                                     {\normalfont\bfseries}}

\numberwithin{equation}{section}

\newcommand{\subalign}[1]{%
  \vcenter{%
    \Let@ \restore@math@cr \default@tag
    \baselineskip\fontdimen10 \scriptfont\tw@
    \advance\baselineskip\fontdimen12 \scriptfont\tw@
    \lineskip\thr@@\fontdimen8 \scriptfont\thr@@
    \lineskiplimit\lineskip
    \ialign{\hfil$\m@th\scriptstyle##$&$\m@th\scriptstyle{}##$\hfil\crcr
      #1\crcr
    }%
  }%
}
\newcommand{\vast}{\bBigg@{4}}
\newcommand{\Vast}{\bBigg@{5}}
\makeatother

\newcommand{\bea}{\begin{eqnarray}}
\newcommand{\eea}{\end{eqnarray}}
\newcommand{\beq}{\begin{equation}}
\newcommand{\eeq}{\end{equation}}
\newcommand{\be}{\begin{equation}}
\newcommand{\ee}{\end{equation}}

\newcommand{\mc}[1]{\mathcal{#1}}

\newcommand{\tmc}[1]{\tilde{\mathcal{#1}}}
\newcommand{\pa}{\partial}

\newcommand{\msg}{\mathsf{g}}

\renewcommand{\title}[1]{\vbox{\center\LARGE{#1}}\vspace{5mm}}
\renewcommand{\author}[1]{\vbox{\center#1}\vspace{5mm}}
\newcommand{\address}[1]{\vbox{\center\footnotesize\em#1}}
\newcommand{\email}[1]{\vbox{\center\footnotesize\tt#1}\vspace{5mm}}

\begin{document}
\begin{titlepage}
 \begin{flushright}

\end{flushright}

\begin{center}

\hfill \\
\hfill \\
\vskip 1cm

\title{Modave lectures on energy conditions in \\ quantum field theory and semi-classical gravity}

\author{Jackson R. Fliss}
\address{
{Department of Applied Mathematics and Theoretical Physics, University of Cambridge,\\ Cambridge CB3 0WA, United Kingdom}\\\vspace{.5 cm}
{Physique Th\'eoretique et Math\'ematique, Universit\'e Libre de Bruxelles \& International Solvay Institutes, CP 231, 1050 Bruxelles, Belgium}
}
  
\email{jackson.fliss@ulb.be}

\end{center}

\vfill
\abstract{
We review well known classical energy conditions and their implications for gravitational solutions, including the celebrated Hawking and Penrose singularity theorems. We then consider quantum fields coupled to gravity, where the topic becomes both richer and more subtle, as even the simplest quantum theories violate local energy conditions. We discuss directions for constraining energy densities in quantum field theories, including averaging over regions of spacetime and bounds relating energy and quantum information. We explore implications of these bounds for quantum field theory coupled to gravity as an effective theory and discuss how they guide our understanding of quantum gravity more broadly. These notes are based on a series of lectures given at the XXI Modave Summer School in Mathematical Sciences. 
}
\vfill

\end{titlepage}

\newpage
{
  \hypersetup{linkcolor=black}
  \tableofcontents
}

\section*{Introduction: A carpenter's guide to semi-classical gravity}\label{sect:intro}

In his essay {\it Physics and Reality} \cite{EINSTEIN1936349}, Albert Einstein characterized the field equation of general relativity as ``a building, one wing of which is made of fine marble, but the other wing is built of low-grade wood":
\beq\tag*{(I.1)}\label{eq:marblewood}
    \underbrace{R_{\mu\nu}-\frac{1}{2}g_{\mu\nu}}_{\text{fine marble}}=\underbrace{8\pi G_N\,T_{\mu\nu}}_{\text{low-grade wood}}~.
\eeq
Within this statement is the frustration of a tensor fixed uniquely by general covariance and capturing the geometry of spacetime itself being contingent upon the stress-energy tensor, which Einstein considered ``a crude substitute for a representation which would correspond to all known properties of matter.'' Without prior knowledge of what restrictions exist on types of matter and how its energy density is composed, \ref{eq:marblewood} is almost vacuous: given one's favorite metric, $g_{\mu\nu}^{\text{favorite}}$ (say with closed timelike curves, a bouncing cosmology, or a traversable wormhole), one could claim it as a solution to the field equation for {\it some} distribution of matter by declaring $T_{\mu\nu}$ as the right-hand side of \ref{eq:marblewood} by fiat.

In these lectures we will take on the role of carpenters, inspecting the quality and shoring up the wooden flank of \ref{eq:marblewood} to provide structure for Einstein's marble wing. That is to say, we will posit reasonable restrictions on how stress-energy is locally composed in physical theories, look for methods to support those assumptions, and, in some cases, even prove those assumptions to be true in known physical theories. We can then explore how those restrictions on stress-energy translate to restrictions on allowed geometries in the form of geometric theorems, the most famous of which are the singularity theorems of Penrose and Hawking \cite{Penrose:1964wq,Hawking:1966sx}.

We will introduce the idea of semi-classical gravity as classical gravity coupled to a quantum source of matter modelled as a quantum field theory. Along the way we will have to make peace with an inconvenient fact about quantum field theories: every quantum field theory violates every local (that is, point-wise) energy condition. We will explore how the notion of averaging energy densities over spacetime regions (and the idea of `quantum interest') allows us to put lower bounds -- quantum energy inequalities -- on how badly energy density can behave in a quantum theory. We will briefly discuss how such inequalities can feed into new Penrose-like and Hawking-like singularity theorems, shoring up the wooden base of \ref{eq:marblewood}, even in the presence of some amount of energy condition violation.

An examination of how energy conditions are violated in quantum field theory will lead us naturally into more modern explorations relating energy to entropy and quantum information more broadly. Such connections run deep into the heart of our understanding of how geometry emerges through holographic duality, and hints at broader lessons for quantum gravity in general. After introducing some basic definitions regarding quantum entanglement, we will describe field theoretic lower bounds on energy density including the celebrated {\it quantum null energy condition} (QNEC). We will then explore how such bounds interplay with semi-classical gravity through quantum focussing and arrive at new, entropic, approaches to constraining geometry in semi-classical and quantum gravity. The culmination of this will be a brief review of Wall's quantum singularity theorem.

Lastly we will finish with a short discussion on the broader context of these lectures with respect to directions of ongoing research.

{\bf NB:} These lecture notes have been influenced by my own pedagogical introduction to the subject matter. I have tried to minimize the intersection with other reviews, however to remain true to what was covered in the summer school where these lectures were given, some of these influences remain evident. In particular Lectures \ref{sect:l1} and \ref{sect:l2} have overlap with the review articles \cite{Kontou:2020bta} and \cite{Fewster:2002ne}. Additionally, as these lectures were being given, \cite{Iizuka:2025xnd} appeared which has coincidental overlap with Lectures \ref{sect:l5} and \ref{sect:l6} of this review.

\section{Lecture 1: Energy conditions in classical field theory}\label{sect:l1}

We will begin with a classical field theory in Lorentzian signature.\footnote{By convention we will work in the ``mostly plus" metric signature.} The action, $S$, of that field theory defines naturally the stress-energy tensor that couples to the Einstein equation through the variational derivative with respect to a background metric
\beq\label{eq:Tcanondef}
    T_{\mu\nu}:=-\frac{2}{\sqrt{-g}}\frac{\delta S}{\delta g^{\mu\nu}}~.
\eeq

We can consider a spacelike Cauchy surface, $\Sigma$. The energy on that surface is given by the Hamiltonian
\beq\label{eq:Ham}
    H=\int_\Sigma\dd^D\vec x~\sqrt{h}~\xi^\mu \xi^\nu\,T_{\mu\nu}~,
\eeq
where $h$ is the induced metric on $\Sigma$ and $\xi^\mu$ is the unit (time-like) normal to $\Sigma$. The simplest classical energy condition is that the total energy is positive:
\beq\label{eq:Hgeq0}
    H\geq 0~.
\eeq
Moreover, when $\xi$ is a Killing vector whose action on the fields also leaves $S$ invariant, $H$ is conserved and independent of the choice of $\Sigma$ along flow of $\xi$.

Positivity and conservation of energy are cornerstones of classical physics and indeed generalize to quantum physics as the positivity and constancy of the eigenvalues of $H$. However conditions on the Hamiltonian itself stands in tension of the experience of a local observer inside a spacetime. In essence we need a `bird's eye' view of $\Sigma$ to measure the total energy. Thus in this lecture we will discuss {\it local} energy conditions and their validity in classical field theories.

To get some intuition for these conditions, it will be useful to consider the perfect fluid ansatz for the stress-energy tensor,
\beq\label{eq:PFT}
    T_{\mu\nu}(x)=\left(\rho(x)+P(x)\right)u_\mu(x)\,u_\nu(x)+P(x)\,g_{\mu\nu}(x)~,
\eeq
where $u^\mu(x)$ is the normalized time-like velocity vector of the fluid ($u^\mu u_\mu=-1$), $\rho(x)$ is its mass density, and $P(x)$ is the isotropic pressure measured in the rest frame of the fluid. Perfect fluid ansatzes are myriad in their applications in modelling the matter sourcing classical spacetimes from stars, black holes, and cosmology \cite{Wald:1984rg}.

\subsubsection*{Weak Energy Condition}
The intuitive pointwise version of \eqref{eq:Hgeq0} is known as the {\it Weak Energy Condition} (WEC) which states that the energy density of an inertial observer, with a normalized worldline tangent vector $t^\mu$, is positive:
\beq\label{eq:WEC}
    t^\mu t^\nu\,T_{\mu\nu}\geq 0~.
\eeq
This condition is very intuitive in terms of the perfect fluid, \eqref{eq:PFT}: the WEC in the rest frame of the fluid states that the mass density of the fluid is positive, 
\beq
    u^\mu u^\nu T_{\mu\nu}(x)=\rho(x)\geq 0
\eeq
while in more general reference frames it implies that the sum of the density and pressure is positive:
\beq
    \rho(x)+P(x)\geq 0~.
\eeq
While intuitive from the perspective of a local observer or from a fluid, the WEC is a rather weak condition geometrically. Its classical geometric implications is the positivity of the Einstein tensor:
\beq
    G_{\mu\nu}\equiv R_{\mu\nu}-\frac{1}{2}R\,g_{\mu\nu}\geq 0~.
\eeq
This condition may seem unnatural, however Feynman gave the following (albeit contrived) physical implication \cite{Feynman:1996kb}. Consider a spatial geodesic ball, $B$, and shrink the ball continuously to a point. The WEC implies that, in this infinitesimal limit, the spatial volume of $B$ is greater than or equal to the spatial volume of a spatial ball in Minkowksi spacetime of the same surface area. 

\subsubsection*{Dominant Energy Condition}

A generalization and strengthening of the WEC is known as the {\it dominant energy condition} (DEC) which states that for any two time-like vectors, $t_1$ and $t_2$,
\beq\label{eq:DEC}
    t_1^\mu t_2^\nu T_{\mu\nu}\geq 0~.
\eeq
The DEC states that in the inertial reference frame set by $t_1$, $-{T^\mu}_\nu\,t_2^\nu$ is a future-pointing causal vector field. This implies that inertial observers always measure energy flux that is causal and future directed. This has the intuitively appealing interpretation of prohibiting superluminal propagation of energy density. The DEC is equivalent to the WEC with the further condition
\beq
    t^\mu t^\nu {T^{\alpha}}_\mu T_{\alpha\nu} \leq 0
\eeq
for any time-like vector, $t$. This implies that in any frame, the energy density is the dominant component of the stress-energy:
\beq
    t^\mu t^\nu T_{\mu\nu}\geq \abs{T_{\mu\nu}}~.
\eeq
In particular, for a perfect fluid this states that
\beq
    \rho\geq\abs{P}~.
\eeq
Much like the WEC, however, there is no immediately clear `marble-side' interpretation of the DEC. 

\subsubsection*{Strong Energy Condition}

For dimensions $d>2$ the {\it strong energy condition} (SEC) can be stated as
\beq\label{eq:SEC}
    t^\mu t^\nu T_{\mu\nu}+\frac{1}{d-2}T\geq 0~,
\eeq
for any normalized time-like vector, $t^\mu$, and where $T=g^{\mu\nu}T_{\mu\nu}$ is the trace of the stress-energy tensor. The `wood-side' physics of the SEC is much less intuitive; in terms of the perfect fluid ansatz it implies
\beq
    (d-3)\rho+(d-2)P\geq 0~,\qquad \rho+P\geq 0~.
\eeq
The quantity on the right-hand side of \eqref{eq:SEC} is sometimes known as the `effective energy density' and the SEC implies that it is positive for all inertial observers. However the `marble-side' physics of the SEC is much more intuitive. It is easy to work out that the Einstein equation implies the curvature condition
\beq\label{eq:TLcurv_cond}
    t^\mu t^\nu R_{\mu\nu}\geq 0~.
\eeq
As we will cover in more detail in the second lecture (Section \ref{sect:l2}), this time-like curvature condition implies that all time-like geodesics converge and is the input for Hawking's singularity theorem.

\subsubsection*{Null Energy Condition}

Lastly we introduce the {\it null energy condition} (NEC) which states that the light-like flux of energy is positivity. That is, for any null vector, $k^\mu$,
\beq\label{eq:NEC}
    k^\mu k^\nu\,T_{\mu\nu}\geq 0~.
\eeq
Of the other classical energy conditions, the NEC is the weakest: for instance, in a perfect fluid it only implies that the sum of the pressure and the density is positive,
\beq
    \rho+P\geq 0~.
\eeq
Regardless, the NEC has several interesting implications. For one, the NEC is relevant for the long distance behavior of asymptotically flat spacetimes. Additionally, the geometric side of the NEC is also clear: the Einstein equation applied to NEC satisfying sources implies the {\it null curvature condition},
\beq\label{eq:null_curv_cond}
    k^\mu k^\nu\,R_{\mu\nu}\geq 0~.
\eeq
As we will soon discuss in more detail, this condition implies that null geodesics are always focussed. This is similar to \eqref{eq:null_curv_cond} in flavor, and features as the key component of the Penrose singularity theorem. However since the NEC is a weaker assumption than the SEC (and thus holds in more situations), this gives the Penrose singularity theorem a stronger starting point.

\subsection*{Example field theories}

Let us now examine the stress tensor in some example field theories to illustrate their (non)compliance with various energy conditions.\\
\\
\textbf{Scalar field}\\
\\
The action of a minimally-coupled real scalar field is given by
\beq\label{eq:scalaract}
    S_\text{scalar}=\int\dd^dx\,\sqrt{-g}\left(-\frac{1}{2}g^{\mu\nu}\pa_\mu\phi\pa_\nu\phi-V(\phi)\right)~,
\eeq
with a corresponding stress-tensor
\beq\label{eq:scalarSET}
    T_{\mu\nu}=\pa_\mu\phi\pa_\nu\phi-g_{\mu\nu}\left(\frac{1}{2}(\pa\phi)^2+V(\phi)\right)~.
\eeq

The WEC quantity is given by
\beq
    t^\mu t^\nu T_{\mu\nu}=\left(t\cdot\pa\phi\right)^2+\frac{1}{2}g^{\mu\nu}\pa_\mu\phi\pa_\nu\phi+V(\phi)~,
\eeq
where we have used that $t^\mu$ is normalized: $g_{\mu\nu}t^\mu t^\nu=-1$. To analyze this further, it is useful to introduce a set of frames, $\{E^\mu_A\}$, and coframes, $\{e^A_\mu\}$, (with $E^\mu_A\,e^B_\mu=\delta^B_A$) such that
\beq
    g_{\mu\nu}=e^A_\mu e^B_\nu\,\eta_{AB}~,\qquad g^{\mu\nu}=E_A^\mu E_B^\nu\,\eta^{AB}~,\qquad \eta=\text{diag}(-1,1,\ldots,1)~.
\eeq
At least in a neighborhood of a point of interest for the WEC, we are free to choose $E_0^\mu=t^\mu$ in which case the WEC quantity is given by
\beq\label{eq:scalarT00}
    t^\mu t^\nu T_{\mu\nu}=\frac{1}{2}E_0(\phi)^2+\frac{1}{2}\delta^{IJ}E_I(\phi)E_J(\phi)+V(\phi)~.
\eeq
This is positive when the potential is positive, $V(\phi)\geq 0$, and the weak energy condition is obeyed.

The strong energy is given by
\beq
    t^\mu t^\nu T_{\mu\nu}+\frac{1}{d-2}T=(t\cdot \pa\phi)^2-\frac{2}{d-2}V(\phi)~.
\eeq
It is easy to see that for positive potentials (e.g. even a simple mass, $V(\phi)=\frac{1}{2}m^2\phi^2$ will do) it is easy to violate the strong energy condition.

For the dominant energy condition, since we have already established that the WEC holds when $V\geq 0$, we simply need to check the quadratic condition,
\beq
    t^\mu t^\nu {T^\alpha}_\mu T_{\alpha\nu}=-2V(\phi)(t\cdot\pa\phi)^2-\left(V+\frac{1}{2}(\pa\phi)^2\right)^2\leq 0~,
\eeq
which for positive potentials is indeed negative and the DEC is satisfied.

Lastly the NEC is easily seen to be obeyed independent of the potential:
\beq
    k^\mu k^\nu T_{\mu\nu}=(k\cdot\pa\phi)^2\geq 0~.
\eeq
\textbf{Yang-Mills fields}\\
\\
Let us consider a non-Abelian vector field valued in an Lie algebra, $\mathfrak{g}$
\beq
    A_\mu=A_\mu^a\,\mathfrak{t}^a~,\qquad\{\mathfrak{t}^a\}\text{ a basis of }\mathfrak g~.
\eeq
The Yang-Mills action is
\beq
S_\text{YM}=-\frac{1}{4\msg_\text{YM}^2}\int \dd^dx\sqrt{-g}\,g^{\mu\nu}g^{\alpha\beta}\tr\left(F_{\mu\alpha}F_{\nu\beta}\right)
\eeq
where the trace is taken in the fundamental representation of $\mathfrak{g}$ (i.e. if $\mathfrak g$ is simple, $\tr(\mathfrak t^a\mathfrak t^b)=\frac{1}{2}\delta^{ab}$ is the Killing form), $F$ is the field strength
\beq
    F^a_{\mu\nu}=\pa_{[\mu}A_{\nu]}^a+\frac{1}{2}{f^a}_{bc}(A^b_\mu A^c_\nu-A^b_\nu A^c_\nu)~,
\eeq
and the $f^{abc}$ are the structure constants of $\mathfrak g$. The stress tensor is computed to be
\beq
    T_{\mu\nu}=\frac{1}{2\msg_\text{YM}^2}\tr\left(g^{\alpha\beta}F_{\mu\alpha}F_{\nu\beta}-\frac{1}{4}g_{\mu\nu}g^{\alpha\beta}g^{\gamma\delta}F_{\alpha\gamma}F_{\beta\delta}\right)~.
\eeq
For investigating the WEC quantity, it is again useful to go to a set of (co)frames comoving with $t^\mu$:
\beq
    t^\mu t^\nu\,T_{\mu\nu}=\frac{1}{4\msg_\text{YM}^2}\left(\delta^{IJ}F_{0I}F_{0J}+\frac{1}{2}\delta^{IJ}\delta^{KL}F_{IK}F_{JL}\right)\geq 0~,\qquad (I,J=1,\ldots, d-1)
\eeq
which is inherently positive. So the WEC is classically obeyed in Yang-Mills theory.

In considering the SEC, it is useful to first look at the trace of the stress-energy tensor
\beq
    T=\frac{4-d}{8\msg_\text{YM}^2}\tr\left(g^{\mu\nu}g^{\alpha\beta}F_{\mu\alpha}F_{\nu\beta}\right)~,
\eeq
noting that it vanishes when $d=4$\footnote{This theory is classically scale invariant in $d=4$. Of course, quantum mechanically the coupling will (famously) run and its renormalization will break this scale invariance.}. The SEC quantity is
\beq
    t^{\mu}t^{\nu}T_{\mu\nu}+\frac{1}{d-2}T=\frac{(d-3)}{(d-2)}\frac{1}{4\msg_\text{YM}^2}\tr\left(\delta^{IJ}F_{0I}F_{0J}\right)+\frac{1}{(d-2)}\frac{1}{8\msg_\text{YM}^2}\tr\left(\delta^{IJ}\delta^{KL}F_{IK}F_{JL}\right)~,
\eeq
which is positive when $d\geq 3$. Since $d=2$ was already excluded from our analysis, unless we do something exotic with the spacetime dimensions (e.g. as in dim-reg) the SEC is also, generally, obeyed in classical Yang-Mills theory.

For the DEC, we need to consider the quadratic quantity which after some massaging takes the following form in a the set of frames adapted to $t^\mu$:
\begin{align}
    t^\mu t^\nu {T^\alpha}_\mu T_{\alpha\nu}=-\frac{1}{16\msg_\text{YM}^4}&\Big(\Big(\delta^{IJ}\tr(F_{0I}F_{0J})+\frac{1}{2}\delta^{IJ}\delta^{KL}\tr(F_{IK}F_{JL})\Big)^2\nonumber\\
    &\qquad-4\delta^{IJ}\delta^{KL}\delta^{MN}\tr(F_{0I}F_{KJ})\tr(F_{0M}F_{LN})\Big)~.
\end{align}
This potentially seems that it could take either sign, however for positive definite Killing forms Cauchy-Schwarz implies
\beq
    \delta^{IJ}\delta^{KL}\delta^{MN}\tr(F_{0I}F_{KJ})\tr(F_{0M}F_{LN})\leq\frac{1}{2}\delta^{IJ}\delta^{KL}\delta^{MN}\tr(F_{0I}F_{0J})\tr(F_{KM}F_{LN})
\eeq
and so
\beq
    t^\mu t^\nu {T^\alpha}_\mu T_{\alpha\nu}\leq-\frac{1}{16\msg_\text{YM}^4}\Big(\delta^{IJ}\tr(F_{0I}F_{0J})-\frac{1}{2}\delta^{IJ}\delta^{KL}\tr(F_{IK}F_{JL})\Big)^2\leq0~,
\eeq
and the DEC is obeyed.

Lastly we can consider the null energy condition
\beq
    k^\mu k^\nu T_{\mu\nu}=\frac{1}{2\msg_\text{YM}^2}\tr\left(k^\mu k^\nu g^{\alpha\beta}F_{\mu\alpha}F_{\nu\beta}\right)=\frac{1}{2\msg_\text{YM}^2}\tr\left(F_{AB}F_{AB}\right)~,\qquad A,B=2,\ldots,d-1~,
\eeq
which is inherently positive. Here we've used a null set of frames with $\{E_\pm,E_A\}_{A=2,\ldots,d-1}$ with $E_+^\mu=k^\mu$. Thus the Yang-Mills theory with a postive Killing form classically obeys all the energy conditions mentioned above.\\
\\
\textbf{The non-minimally coupled scalar field}\\
\\
As a final example we reconsider the scalar field \eqref{eq:scalaract} and add to it a coupling directly to the Ricci scalar
\beq
    S_\xi=-\frac{\xi}{2}\int \dd^dx\sqrt{-g}\,R\,\phi^2~.
\eeq
Such a coupling might naturally arise say in a dimensional reduction of pure gravity, or as a coupling in effective field theory treatment of gravity. Additionally such a coupling is necessary to couple the conformal scalar to a curved background in which case the coupling takes the specific value (known as the {\it conformal coupling}):
\beq
    \xi_\text{c}=\frac{d-2}{4(d-1)}~.
\eeq
At this coupling the action is classically invariant under $(g,\phi)\rightarrow (\Omega^2\,g,\Omega^{\Delta_\phi}\phi)$ with $\Delta_\phi=\frac{d-2}{2}$ the classical scaling dimension of a free scalar field. This coupling has an interesting effect on the stress-energy tensor in that, in its canonical definition, \eqref{eq:Tcanondef}, it depends explicitly on the Einstein tensor:
\beq\label{eq:NMCT}
    T_{\mu\nu}=\pa_\mu\phi\pa_\nu\phi-g_{\mu\nu}\left(\frac{1}{2}(\pa\phi)^2+V\right)-\xi\left(g_{\mu\nu}\nabla^2+\nabla_\mu\nabla_\nu-G_{\mu\nu}\right)\phi^2~.
\eeq
It is sometimes common to utilize the Einstein equation
\beq
    G_{\mu\nu}=8\pi G_N T_{\mu\nu}~,
\eeq
and move all the ${G_{\mu\nu}}'s$ to one side and write an effective stress-tensor
\beq
    G_{\mu\nu}=8\pi G_N\,T_{\text{eff},\mu\nu}~,
\eeq
which only depends on field-theory quantities:
\beq
    T_{\text{eff},\mu\nu}=\frac{1}{1-8\pi G_N\xi\phi^2}\left(\pa_\mu\phi\pa_\nu\phi-g_{\mu\nu}\left(\frac{1}{2}(\pa\phi)^2+V\right)-\xi\left(g_{\mu\nu}\nabla^2+\nabla_\mu\nabla_\nu\right)\phi^2\right)~.
\eeq
It is this latter stress-tensor that is relevant for constraining semi-classical geometries. Although $S_\xi\rightarrow 0$ in Minkowski space, the non-minimally coupled stress-energy tensor does not reduce to the minimally coupled stress-energy tensor in this limit:
\beq
    T_{\mu\nu}\Big|_\text{Mink.}=\pa_\mu\phi\pa_\nu\phi-\eta_{\mu\nu}\left(\frac{1}{2}(\pa\phi)^2+V\right)-\xi\left(\eta_{\mu\nu}\pa^2+\pa_\mu\pa_\nu\right)\phi^2~.
\eeq
which reflects that $T$ knows about variations about a background as well as a background itself. In this case the terms proportional to $\xi$ are actually {\it improvement terms} to the canonical stress tensor. We also note that when $\xi=\xi_\text{c}$ and $V=0$, these improvement terms make $T_{\mu\nu}$ traceless $T=0$ which is important for its interpretation as a conformal stress tensor.

We saw that the minimally coupled scalar field already violates the SEC (at least when $V\geq0$); here the non-minimal coupling spoils every other energy condition. Let us just investigate the NEC which is the weakest of all the conditions:
\beq\label{eq:NMCTkk}
    k^\mu k^\nu T_{\mu\nu}=(1-2\xi)(k\cdot\pa\phi)^2-2\xi\,\phi\left(k^\mu k^\nu\nabla_\mu \nabla_\nu+\frac{1}{2}R_{kk}\right)\phi~,
\eeq
which can be potentially negative. Additionally the effective null energy density can possibly be negative
\beq
    k^\mu k^\nu\,T_{\text{eff},\mu\nu}=\frac{1}{1-8\pi G_N\xi\phi^2}\left((1-2\xi)(k\cdot\pa\phi)^2-2\xi\phi\left(k^\mu k^\nu \nabla_\mu\nabla_\nu+\frac{1}{2}R_{kk}\right)\phi\right)~.
\eeq
and can even switch sign for trans-Planckian field values, $\phi^2>(8\pi G_N\xi)^{-1}$.

In summary, none of the above classical energy conditions are immutable: each has a counterexample in scalar field theory. Regardless, the NEC stands out as the classical energy condition that holds `the most often'\footnote{In fact, we will see in Lecture \ref{sect:l3} that the regimes of field space of the non-minimally coupled scalar which violate the NEC has somewhat pathological features that call into question trusting them within effective field theory \cite{Fliss:2023rzi}.}. We will see in the next lecture that the NEC plays a special role in geometric theorems as well.

\section{Lecture 2: Constraining classical gravity}\label{sect:l2}

In this lecture we are going to use the energy conditions in the previous lecture to say something about the `marble side' of \ref{eq:marblewood} and the role of energy conditions in {\it singularity theorems} in general relativity.

\subsection*{Singularity theorems in GR}

We should be begin by stating more precisely what we mean by a singularity in general relativity. As physicists we are often accustomed to saying that a singularity is a place where the mathematics of general relativity breaks down. While true in spirit, this is not very precise. For instance, the Schwarzschild black hole metric components
\beq
    \dd s^2=-f(r)\dd t^2+f(r)^{-1}\dd r^2+r^2\dd\Omega^2~,\qquad f(r)=\left(1-\frac{r_\text s}{r}\right)~,
\eeq
`breaks down' when $r=r_\text s$, however we are all very aware that this is simply the signal of a horizon which, locally, is a very ordinary place. This is a coordinate singularity that can be removed by going to, say, Kruskal coordinates. More serious is the singularity at $r=0$ which can be stated in a coordinate invariant manner as a curvature divergence, e.g.
\beq
    R^{\mu\nu\rho\sigma}R_{\mu\nu\rho\sigma}=\frac{48\,G_N^2\,M^2}{r^6}~,
\eeq
in $d=4.$  

We might use this example to state that singularities are places where spacetime curvatures blow up and this is not a bad definition, particularly when thinking of a GR as an effective theory organized by a curvature expansion. However it will be much more helpful for us to relax our definition of singularity to something more mathematically robust. For the purposes of these lectures, we will define a singularity as the following:\\
\\
\textit{A spacetime possesses a singularity if it possesses at least a single incomplete and inextendible causal geodesic.}
\\

This broad definition of `singularity,' on the one hand, does not indicate what kind of singularity a spacetime possesses. While this is certainly a disadvantage, it comes with the advantage of encapsulating a larger class of singular spacetimes that are harder to diagnose with curvature invariants (e.g. conical singularities whose curvature invariants can remain small everywhere except a single point).\\
\\
\textbf{Why singularity theorems?}\\
\\
Before we move into a more detailed description of singularity theorems, let us motivate why we should care. For one, as a physicist reading lecture notes on semi-classical gravity, chances are you are personally interested, as am I, in the large-scale structure of our universe and amazed at the ways we diagnose its features using the tools we mere mortal mammals have access to. More historically, the discovery of the singularities at $r=0$ and $r=r_\text{s}$ were initially quite startling and their interpretation remained unclear (indeed, it over a decade after Schwarzschild's solution to even realize that the horizon singularity was a coordinate artifact). General relativity is a notoriously difficult system to solve, involving a set of coupled, non-linear partial differential equations. Exact solutions are rare and employ a large amount of symmetry. This calls into question the physicality of curvature singularities even at the classical level: perhaps these are pathologies of symmetric spacetimes that are resolved by perturbations breaking that symmetry. For instance, could it not happen, in the absence of spherical symmetry, that infalling geodesics narrowly `miss' and bounce away? Do we have to solve the non-linear equations to find out?

The result of singularity theorems of this lecture is to answer firmly that this is not the case: spacetime singularities are robust and, given a reasonable set of initial data and a motivated energy condition, they are generic. Thus if your prejudice is that spacetime singularities are unphysical and should be resolved, it will not happen classically and instead will require quantum effects. Later lectures (Lectures \ref{sect:l4} and \ref{sect:l6}), will give indication that even QFT effects are not enough and such a resolution will require quantum gravity effects at a higher energy scale (such as the string scale or the Planck-scale).

\subsubsection*{Singularity theorems: general structure}

The general structure of the singluarity theorems of this section follow the basic recipe of proving the incompatibility of a set of assumptions on the initial condition, the global spacetime structure (such as the absence of singularities), the validity of GR, and a motivated energy condition:

\textit{\begin{itemize}
    \item A surface, $\Sigma$, satisfying a global condition (e.g. compactness, non-compactness, Cauchy) and an `initial condition' (e.g. initial convergence of null geodesics, divergence of normals).
    \item Nonsingularity and/or completeness of the spacetime manifold (e.g. null or timelike geodesic completeness).
    \item Satisfaction of the Einstein equation.
    \item Satisfaction of an energy condition (e.g. the WEC, the SEC, or the NEC).
\end{itemize}}
The goal is then to show that these ingredients are mutually incompatible. Barring either a drastic modification of Einstein gravity as an effective theory or the violation of an energy condition (more on this later!), then the conclusion is that initial conditions of the first bullet inevitably lead to violation of the second, i.e. geodesic incompleteness.

\subsubsection*{The Raychaudhuri equation}

The key tool utilized in the following singularity theorems is a geometric identity known as the {\it Raychaudhuri equation} (this is actually a pair of identities depending on if we are interested in timelike geodesics or null geodesics). Since this identity will continue to play a central role in later lectures we pause to make special note of it here. We will follow closely the derivations found in Wald \cite{Wald:1984rg}.

In a nutshell, the Raychaudhuri equation succinctly spells out the conditions by which geodesics converge or diverge depending on background curvature. More specifically, it spells out the behaviour of infinitesimal area elements as they evolve under proper time (or null affine time).\\
\\
\textbf{Timelike Raychaudhuri}\\
\\
Let $\Sigma$ be a $(d-1)$ spacelike submanifold and $t$ be a timelike vector field generating a congruence of timelike geodesics emanating orthogonally from $\Sigma$. The induced metric on $\Sigma$ is given by $h_{\mu\nu}=g_{\mu\nu}+t_\mu t_\nu$. We can decompose the divergence of $t^\mu$ into its trace, traceless-symmetric, and antisymmetric pieces as
\beq
    \nabla_\mu t_\nu=\frac{1}{(d-1)}\vartheta+\sigma_{\mu\nu}+\omega_{\mu\nu}~,
\eeq
where
\beq
    \vartheta=h^{\mu\nu}\nabla_\mu t_\nu~,\qquad \sigma_{\mu\nu}=\nabla_{(\mu}t_{\nu)}-\frac{1}{(d-1)}\vartheta\,h_{\mu\nu}~,\qquad \omega_{\mu\nu}=\nabla_{[\mu}t_{\nu]}~,
\eeq
are the {\it expansion}, the {\it shear}, and the {\it rotation}, respectively. In words these quantify how much an infinitesimal unit of area on $\Sigma$ expands, squeezes, and twists under the flow of $t^\mu$. It is also useful to note that we can express the expansion as 
\beq
    \vartheta=\frac{1}{\sqrt{h}}\frac{\dd}{\dd\tau}\sqrt{h}~,\qquad \frac{\dd}{\dd\tau}:=t^\mu\nabla_\mu~,
\eeq
which makes that statement a bit more manifest. The Raychaudhuri equation itself follows from taking the second flow derivative, $\frac{\dd}{\dd\tau}=t^\mu\nabla_\mu$, of the above equations and massaging (being sure to use $t^\mu t_\mu=-1$, that it is geodesic, and $\nabla_{[\mu}\nabla_{\nu]}=R_{\mu\nu}$):
\beq\label{eq:TLRayEq}
    \frac{\dd\vartheta}{\dd\tau}=-\frac{1}{(d-1)}\vartheta^2-\sigma_{\mu\nu}\sigma^{\mu\nu}+\omega_{\mu\nu}\omega^{\mu\nu}-R_{\mu\nu}t^\mu t^\nu~.
\eeq
This equation relates that the rate of change of the expansion itself to how much it shears, rotates, and to the background curvature. In particular for geodesic flows orthogonal to $\Sigma$, it is possible to arrange for them to be irrotational ($\omega_{\mu\nu}$). If additionally the {\it weak curvature condition} is satisfied, $R_{\mu\nu}t^\mu t^\nu\geq0$, then the Raychaudhuri equation implies
\beq
    \frac{\dd\vartheta}{\dd\tau}+\frac{1}{d-1}\vartheta^2\leq 0~,\qquad(t^\mu t^\nu R_{\mu\nu}\geq0)~,
\eeq
which formalizes the intuition that gravity brings free-falling massive bodies together.\\
\\
\textbf{Null Raychaudhuri}\\
\\
Due to the somewhat singular nature of null submanifolds, formulating the Raychaudhuri equation for null geodesics requires some minor retooling. Let $\Gamma$ be a $(d-2)$-dim spacelike submanifold. Emanating from every point on $\Gamma$ are two independent future pointing null vectors which we can collect into two vector fields in a neighborhood surrounding $\Gamma$ which we will denote $k^\mu$ and $\ell^\mu$ (note that $k^\mu k_\mu=\ell^\mu\ell_\mu=0$ and we normalize $k^\mu\ell_\mu=-1$). We can locally let these vector fields generate a set of null geodesic congruences, e.g.
\beq
    k^\mu \nabla_\mu k_\nu=\ell^\mu \nabla_\mu \ell_\nu=0~.
\eeq
The induced metric on $\Gamma$ is given by
\beq
    \hat h_{\mu\nu}=g_{\mu\nu}+\ell_\mu k_\nu+\ell_\nu k_\mu~.
\eeq
We will additionally define
\beq
    B_{\mu\nu}=\nabla_\mu k_\nu~,
\eeq
and `project it down' to $\Gamma$ as 
\beq
    \hat B_{\mu\nu}\equiv \nabla_\nu k_\mu+(\ell^\rho\,\nabla_\nu k_\rho)k_\mu+(\ell^\rho\,\nabla_\rho k_\mu)k_\nu~,
\eeq
which we can verify is annihilated by both contraction with $\ell^\mu$ and $k^\mu$. We now consider its trace, traceless, and anti-symmetric decomposition
\beq
    \theta =\hat h^{\mu\nu}\hat B_{\mu\nu}~,\qquad\hat\sigma_{\mu\nu}=\hat B_{(\mu\nu)}-\frac{1}{(d-2)}\theta\,\hat h_{\mu\nu}~,\qquad
    \hat\omega_{\mu\nu}=\hat B_{[\mu\nu]}~.
\eeq
These objects have the same interpretation as before but now applied to the codimension-2 area elements of $\Gamma$. Namely we can also show,
\beq
    \theta=\frac{1}{\sqrt{\hat h}}\frac{\dd}{\dd\lambda}\sqrt{\hat h}:=\frac{1}{\sqrt{\hat h}}k^\mu\nabla_\mu\sqrt{\hat h}~.
\eeq
We can now consider the null evolution of $B_{\mu\nu}$ (unhatted), $\frac{\dd}{\dd\lambda}B_{\mu\nu}:=k^{\rho}\nabla_\rho B_{\mu\nu}$, massage it, and then project it down to $\Gamma$, $\widehat{\frac{\dd}{\dd\lambda}B_{\mu\nu}}$, to find
\beq\label{eq:nullRayEq}
    \frac{\dd\theta}{\dd\lambda}=-\frac{1}{(d-2)}\theta^2-\hat\sigma_{\mu\nu}\hat\sigma^{\mu\nu}+\hat\omega_{\mu\nu}\hat\omega^{\mu\nu}-R_{\mu\nu}k^\mu k^\nu~.
\eeq
The key feature of the null Raychaudhuri equation is that when null curvature is positive ($R_{\mu\nu}k^\mu k^\nu\geq 0$), irrotational flows ($\hat\omega=0$) always have non-increasing expansions:
\beq
    \frac{\dd\theta}{\dd\lambda}+\frac{1}{(d-2)}\theta^2\leq 0~,\qquad (R_{kk}\geq0)~,
\eeq
which is the statement that gravity {\it focusses} light rays.

The null Raychaudhuri equation will play a special role in what follows partially because of the relative nicety of the NEC. It is also interesting to note that the areas defining the null expansion are codimension-2; this suggests a very speculative connection to entanglement entropy QFT where the codimension-2 area of the entangling surface plays a special role. We will make this connection more solid in Lectures \ref{sect:l5} and \ref{sect:l6}. Lastly, null limits play a special role in QFT: quantization on null surfaces simplifies in super-renormalizable field theories (a point we will return to in the next lecture) and CFT OPEs simplify significantly at null separation, essentially being organized by twist as opposed to conformal dimension.

\subsection*{The Penrose Singularity Theorem}

With the general ethos, structure, and the key tools of classical singularity theorems laid out, let us now take a look at our first example singularity theorem (and indeed, arguably, the most famous one) which was developed by Penrose in 1965 in a beautiful and short paper \cite{Penrose:1964wq}.

Penrose considers a situation in which a star starts to collapse, but without an assumption of spherical symmetry. He notes that shortly after collapse and the formation of a horizon there is a spacelike region just inside the horizon whose boundary is ${\it trapped}$:\\
\\
{\it A trapped surface is a closed codimension-2 spacelike submanifold whose two future-directed null geodesic congruences have negative expansion.}\\
\\
More plainly, future directed light-rays shot from a trapped surface are initially converging. This is physically sensible for the codimension-2 surface sitting just inside horizon: light cannot escape and so any future-directed ray converges inwards.

Penrose then shows that the NEC is incompatible with null geodesic completeness. We will state his theorem formally.

\begin{theorem}[Penrose Singularity Theorem]\label{thm:PST}
Let $\Sigma$ be an initial non-compact hypersurface such that (i) every past-directed causal geodesic can be extended to intersect $\Sigma$ (i.e. $\Sigma$ is Cauchy),(ii) the future development of $\Sigma$, $M_+$ is a nonsingular manifold, (iii) $M_+$ contains a trapped surface, $\Gamma$, (iv) every future-directed null geodesic of $M_+$ can be extended to arbitarily large affine parameter (i.e. null geodesic completeness), and (v) the null curvature is positive. Then (i)-(v) are incompatible.
\end{theorem}

Given what we covered in Lecture \ref{sect:l1}, it is clear that assumption {\it (v)} is equivalent to the supposition of both the Einstein equation and the NEC. Even though Penrose's paper is short and understandably written will only give a coarse rundown of his proof in this lecture.

\begin{proof}[Proof (sketch)]
    Because $\Gamma$ is trapped, both systems of future-directed null congruences emanating from $\Gamma$ are initially negative, $\theta^{(1,2)}(\lambda=0):=\theta^{(1,2)}_0<0$. The null Raychaudhuri equation, \eqref{eq:nullRayEq}, and the null curvature condition applied to each system of congruences imply that the corresponding expansion diverges in finite affine parameter
    \beq
        \theta^{(i)}(\lambda)\leq-\frac{|\theta_0^{(i)}|}{1-\alpha\lambda}~,\qquad \alpha=\frac{|\theta_0^{(i)}|}{(d-2)}~,
    \eeq
    and so the congruence develops a caustic in finite affine parameter. The future lightcone of $\Gamma$, call it $\mc L$, is then compact. However since $\Sigma$ is Cauchy and $M_+$ is complete, the past causal geodesics emanating from $\mc L$ (or a spacelike surface approximating $\mc L$) establish a one-one map from $\mc L$ to $\Sigma$. This is impossible due to non-compactness of $\Sigma$.
\end{proof}

Note that the non-compactness of $\Sigma$ is a key component of Penrose's theorem. This is suited for establishing robustness of singularities under gravitational collapse in asymptotically flat spacetimes. However there are several situations, such as expanding or cosmological, spacetimes that have compact Cauchy surfaces. Thus Penrose's theorem is not sufficient for establishing the robustness of say, the Big Bang singularity going backwards in time. One year after Penrose, Hawking addressed exactly this scenario.

\subsection*{The Hawking Singularity Theorem}

\begin{theorem}[Hawking Singularity Theorem]
Let $M$ be a spacetime such that (i) there exists a compact spacelike Cauchy surface, $\Sigma$, of $M$, (ii) the timelike unit normal vector field to $\Sigma$, $t^\mu$, has negative expansion, $\vartheta<0$, for all points in $\Sigma$, (iii) $M$ is future timelike geodesically complete, and (iv) the strong curvature condition $R_{\mu\nu}t^\mu t^\nu\geq 0$, holds for all points in $M$. Then (i)-(iv) are mutually incompatible.
\end{theorem}

Two comments: First, assumption (iv) is obviously equivalent to assuming both the Einstein equation and the SEC. Second, the time reverse of this theorem is that if the past-directed expansion is negative everywhere on $\Sigma$ then the strong energy condition implies that the spacetime is past timelike geodesically incomplete. A sketch of the proof is as follows.

\begin{proof}[Proof (sketch)]
Because $\Sigma$ is compact, the expansion reaches a maximum, $\bar\vartheta$, which by supposition is negative. So all flows of $t^\mu$ terminate in finite proper time at least before, $\bar\tau$, the time that $\bar\vartheta$ diverges. Thus the future causal development of $\Sigma$, $D_+(\Sigma)$ is finite, with every point $p\in D_+(\Sigma)$ being within a proper time $\bar\tau$ of a point in $\Sigma$. Let $\gamma$ be a future directed timelike curve (not necessarily a geodesic generated by $t^\mu$). If $\gamma$ were extendable then it would have to exit the causal development of $\Sigma$ which contradicts it being a Cauchy surface of $M$. Thus all future directed timelike curves are incomplete and inextendable. 
\end{proof}

There have been subsequent extensions to these singularity theorems, jointly by Hawking and Penrose \cite{Hawking:1970zqf}, weakening the assumptions of the Hawking Singularity Theorem to requiring the WEC as opposed to the SEC. However the main ideas of these singularity theorems are essentially the same as what we have described above.

\section{Lecture 3: Energy conditions in quantum field theory}\label{sect:l3}
  
We now move from classical field theory to quantum field theory. The motivation is obvious: 
neglecting gravity, our most successful theory describing matter is the Standard Model of particle physics based on quantum field theory. Moreover we might expect there to be a regime where we can consistently treat the Einstein equation {\it semi-classically}, i.e. we replace the wood side with an expectation value of the stress tensor
\beq\label{eq:SCEineq}
    G_{\mu\nu}=8\pi G_N\langle T_{\mu\nu}\rangle_\psi~.
\eeq

There are lots of reasons to be suspect of the above equation.\footnote{A much more consistent approach is to start with the gravity and matter path-integral and integrate out the matter to arrive at a gravitational effective action, schematically:
\beq
    Z_\text{matter+grav}=\int Dg_{\mu\nu}D\Phi \,e^{iS_\text{EH}[g]+iS_\text{matter}[g,\Phi]}=\int Dg_{\mu\nu}e^{i\Gamma_\text{grav,eff}[g]}~.
\eeq
The saddle-point equations of $\Gamma_\text{grav,eff}$, if they exist, are then a consistent semi-classical description of the system.}
For one, even defining it requires us to quantize field theory in generic curved backgrounds which, despite being a rich subject that we will not jump into here, is still subtle even for free field theories. Another confusion lies in which how we choose the state, $\psi$, to define the right-hand side of \eqref{eq:SCEineq}. Despite its appearance, \eqref{eq:SCEineq} is not even covariant: at best \eqref{eq:SCEineq} can only hold on a time-slice since the vacuums of a QFT can be significantly different in different coordinate patches (a famous example being the Unruh radiation encountered in the Rindler patch of Minkowski space). More alarming is the mismatch of the classicality of left-hand and right-hand sides of \eqref{eq:SCEineq} which is inconsistent when $\psi$ is a quantum superposition of large gravitating objects \cite{Page:1981aj}.

Nevertheless, semi-classical gravity described as above has its utilities and should not be discarded. It provides a framework for organizing perturbative backreaction: one quantizes a field theory on a classical vacuum background (say a cosmology or the exterior of a black hole), and then computes the effect of adding a quantum source by solving \eqref{eq:SCEineq} order by order in a $G_N$ expansion of $\langle T\rangle$ and the background metric $g_{\mu\nu}$. 

For our purposes we will simply assume that we have in a hand {\it a consistent solution} to \eqref{eq:SCEineq}, e.g. there exists a state defined on a Cauchy slice quantized in the curved background of $g_{\mu\nu}$ such that the stress-energy tensor expectation value in that state satisfies \eqref{eq:SCEineq}. If this sounds impractical or even a bit circular, you are not entirely wrong. However the consistency of \eqref{eq:SCEineq} can be checked a bit more explicitly in AdS/CFT where the stress-energy tensor expectation value can be related to a CFT stress-energy tensor expectation value which itself is a boundary value of the metric. More broadly, we might still hope that general features and constraints on energy in QFT may still constrain geometric solutions to \eqref{eq:SCEineq}. This is the ethos that we will push in these lectures.

\subsection*{Failure of energy conditions in QFT}

Having covered several classical energy conditions and having written our apologia for the semi-classical Einstein equation on the basis that maybe it is constrained by general features of stress-energy in quantum field theory, let us now show that every such pointwise energy condition from Lecture \ref{sect:l1} is violated in quantum field theory!

Let us first give a simple example in scalar field theory on Minkowski spacetime.
\subsubsection*{Example: the 0+2 state in scalar field theory}

Consider the massive scalar field in Fock quantization in Minkowski spacetime:
\beq\label{eq:scalarquant}
    \hat \phi(t,\vec x)=\int\frac{\dd^{d-1}\vec k}{(2\pi)^{d-1}\sqrt{2\omega_{\vec k}}}\left(a_{\vec k}~e^{i\omega_{\vec k}t+i\vec k\cdot x}+a^\dagger_{\vec k}~e^{-i\omega_{\vec k}t-i\vec k\cdot\vec x}\right)~,\qquad \omega_{\vec k}=\sqrt{\vec k^2+m^2}~.
\eeq
with commutation relations
\beq
    [a_{\vec k},a^\dagger_{\vec k'}]=(2\pi)^{d-1}\delta^{(d-1)}(\vec k-\vec k')~.
\eeq
Of course the stress-energy tensor \eqref{eq:scalarSET} involves a coincident field product over field operators, $\hat\phi$, which is undefined in quantum field theory. Thus we supplement this with a normal ordering prescription, denoted by $:~~:$, that in the Fock quantization is canonically taken to be putting all the ${\hat a^{\dagger}}$'s to the left and all the ${\hat a}$'s to the right:
\beq
    \hat T_{\mu\nu}=:\pa_\mu\phi\pa_\nu\phi:-\frac{1}{2}\eta_{\mu\nu}\left(:(\pa\phi)^2:+m^2:\phi^2:\right)~.
\eeq
For the following example we can consider the energy density appearing in the WEC
\begin{align}
    \hat T_{00}=&\int_{\vec k_1}\int_{\vec k_2}\left(\omega_{\vec k_1}\omega_{\vec k_2}+\vec k_1\cdot\vec k_2+m^2\right)a^\dagger_{\vec k_1}a_{\vec k_2}e^{-i(k_1-k_2)_\mu x^\mu}\nonumber\\
    &-\frac{1}{2}\int_{\vec k_1}\int_{\vec k_2}(\omega_{\vec k_1}\omega_{\vec k_2}+\vec k_1\cdot\vec k_2-m^2)\left(a_{\vec k_1}a_{\vec k_2}e^{i(k_1+k_2)_\mu x^\mu}+a^\dagger_{\vec k_1}a^\dagger_{\vec k_2}e^{-i(k_1+k_2)_\mu x^\mu}\right)
\end{align}
where we have introduced a shorthand $\int_{\vec k}:=\int\frac{\dd^{d-1}\vec k}{(2\pi)^{d-1}\sqrt{2\omega_{\vec k}}}$. The first line, the number conserving term, is inherently positive while the second line can have either sign depending on the state. This second term can come into play in expectation values of states that are superpositions of particle number. 

For example we consider the ``0+2'' state \cite{Pfenning:1998ua}:
\beq
    |\psi_\alpha\rangle =\cos\alpha|\Omega\rangle+\frac{\sin\alpha}{\sqrt{2}V}a^\dagger_{\vec k}a^\dagger_{\vec k}|\Omega\rangle~.
\eeq
Above $V\sim\delta^{(d-1)}(\vec k=0)$ is a stand-in for an IR regulator on the spatial volume. This can be made more rigorous by giving the particle momentum a small Gaussian width and the reader is encouraged to repeat this exercise to convince themselves that this detail is inessential to our argument below. We evaluate the $T_{00}$ in this state to find
\beq
    \langle \hat T_{00}(t,\vec x)\rangle_{\psi_\alpha}=\frac{2}{V}\left(\omega_{\vec k}\sin^2\alpha-\sqrt{2}\sin\alpha|\vec k|^2\cos\left(\omega_{\vec k}t+\vec k\cdot\vec x\right)\right)
\eeq
The first term arises from the number conserving piece of the stress-energy tensor and is indeed positive, however the second term oscillates between signs. For sufficiently small $\alpha$, this term can win over the the positive term and we find regions of negative energy density.

\subsubsection*{A general argument against positive energy densities}

One may skeptically brush aside the above example as a peculiarity of free scalar field theory, or because we have evaluated the stress-energy tensor at a point: indeed energy densities localized to a point are rarely meaningful and we are usually taught in any introductory course on QFT that expectation values are distributional i.e. should be integrated against test functions. Unfortunately (or perhaps interestingly) the persistence of negative energy densities is a general feature of local quantum field theories, even if integrated against a smooth function compactly supported on a region of a Cauchy slice. As far as I know, the clearest argument of this was given by Epstein, Glaser, and Jaffe in 1965 \cite{Epstein:1965zza}.

\begin{theorem}[Epstein, Glaser, Jaffe (1965)]
Consider a classically positive operator, $\mc O(x)\geq 0$, quantized $\mc O\rightarrow\hat{\mc O}$ and renormalized via the standard normal ordering prescription:
\beq
    :\hat{\mc O}:=\hat {\mc O}-\langle \hat {\mc O}\rangle_\Omega~,
\eeq
such that $\langle :\hat {\mc O}:\rangle_\Omega=0$. We allow ourselves ourselves to integrate $:\hat{\mc O}:$ against a smooth test function, $h$, of compact support on $\Sigma$:
\beq
    :\hat{\mc O}_h:=\int_\Sigma\dd^{d-1}\vec x\,h(\vec x):\hat{\mc O}(\vec x):~.
\eeq
Either there exists a state such that $:\hat{\mc O}_h:$ possesses a negative expectation value, or $:\hat{\mc O}_h:$ is zero as an operator.
\end{theorem}

\begin{proof}
    Let us show firstly that either $:\hat{\mc O}_h:$ annihilates the vacuum, $:\hat{\mc O}_h:$ or possesses a negative expectation value. Suppose it does not annihilate the vacuum. Then there exists at least one state, $|\phi\rangle$, such that
    \beq
        \langle \phi|:\hat{\mc O}_h:|\Omega\rangle :=\mathfrak O_h\neq 0~.
    \eeq
    Now consider the 2x2 matrix spanned by expectaion values of $:\hat{\mc O}_h:$ between $|\Omega\rangle$ and $|\phi\rangle$:
    \beq
        M:=\left(\begin{array}{cc}0&\mathfrak O_h\\\mathfrak O^\ast_h &\langle :\hat{\mc O}_h:\rangle_\phi\end{array}\right)~.
    \eeq
    The determinant of $M$ is negative and so contains one positive and one negative eigenvalue. It follows that there exists a linear combination of $|\Omega\rangle$ and $|\phi\rangle$ such that $:\hat{\mc O}_h:$ has a negative expectation value.

    So if $:\hat{\mc O}_h:$ is positive, it must annihilate the vacuum exactly, $:\hat{\mc O}_h:|\Omega\rangle=0$. However this implies that $:\hat{\mc O}_h:$ is zero as an operator. This is because by causality
    \beq
        [:\hat{\mc O}_h:,\hat{\mc O}'(\vec y)]=0~,
    \eeq
    for any other operator spacelike separated from $:\hat{\mc O}_h:$, $\vec y\notin\text{supp}(h)$. This implies
    \beq
        :\hat{\mc O}_h:\hat{\mc O}'(\vec y)|\Omega\rangle = 0~.
    \eeq
    However by the Reeh-Schlieder theorem \cite{Reeh:1961ujh} such states of local operators acting on the vacuum are dense in the Hilbert space and so $:\hat{\mc O}_h:$ must be the zero operator.
\end{proof}

Note how this proof generalizes our earlier example: the two-particle state is precisely the state with non-zero overlap with $\hat T_{00}$ acting on the vacuum and our negative energy density was constructed out of the superposition of these states. It is also useful to consider how genuinely positive quantum operators, such as the Hamiltonian or total momentum evade the above theorem: they can be possible because they annihilate the vacuum. This does not imply they are exactly zero because they are integrated over an entire Cauchy slice and so there is no analogous argument using the Reeh-Schlieder theorem. Lastly, we remark that the Reeh-Schlieder theorem follows from the dense entanglement structure of the QFT vacuum and so this is our first encounter on the interesting interplay between energy and quantum information which we will revisit in Lectures \ref{sect:l5} and \ref{sect:l6}.

\subsection*{QEIs}\label{sec:QEIs}

Above we have seen that quantum fields can yield negative expectation values even for classically positive operators due to the necessity of renormalizing their short distance divergences. This means that the pointwise energy conditions of Lecture \ref{sect:l1} are all violated. The situation is in fact, even worse: typically not only are these expectation values possibly negative, they are not even lower bounded. The argument from Witten is the following \cite{WittenTalk}: imagine smearing a classically positive operator in some compact domain and shrinking the domain to a point. Initially it will be negative for some state. When that domain is very small, that expectation value is dominated by high-frequency modes and we can better and better approximate the expectation value by the expectation value at the UV fixed point which will be a CFT and display a scaling symmetry. For this purposes we will assume our operator of interest approaches a CFT primary operator (which is true for the stress-energy tensor) and the state of interest approaches a sum over states prepared by primary operators which is a dense set in a CFT via the state-operator correspondence. The expectation value of a primary operator obeys a scaling symmetry
\beq
    \langle:\hat{\mc O}(\lambda x)\rangle_{\psi}\Big|_{\text{high-freq.}}\approx \lambda^{-\Delta}\langle:\hat{\mc O}(x)\rangle_{\tilde\psi}\Big|_\text{high freq.}~,
\eeq
where $\tilde\psi$ is the state transformed by the dilatation operator. If the expectation value was negative in the state $\psi$ this implies it is always possible to find a state of arbitrarily negative expectation values by shrinking the domain upon which $\hat{\mc O}$ has support.

Thus to get physically meaningful quantities, we need to `smear' operators over some spacetime domain. Given our discussion in the previous section of this lecture, such smeared operators can still be negative, but their negativity is under control due to the smearing. This brings us to the introduction of {\it quantum energy inequalities} (QEIs). We will start this discussion with a simple example of a QEI.

\subsubsection*{Example: Worldline QEI of a 4$d$ scalar field}

Consider again the scalar field energy density, $T_{00}$, in Minkowski space in four dimensions. We average the energy density along a timelike worldline with a smooth function of compact support $g(t)$. This is lower bounded by \cite{Fewster:1998pu}
\beq\label{eq:FewsterEvesonBound}
    \int \dd t\,g(t)^2\langle T_{00}(t,\vec x)\rangle_\psi\geq - \frac{1}{16\pi^3}\int_{m}^\infty \dd\omega\,\omega^4\,\abs{\tilde g(\omega)}^2\,Q_3(\omega/m)~,
\eeq
with
\beq
    Q_3(x)=\sqrt{1-\frac{1}{x^2}}\left(1-\frac{1}{2x^2}\right)-\frac{1}{2x^4}\log\left(x+\sqrt{x^2-1}\right)~.
\eeq
for all Hadamard states, $\psi$.\footnote{A Hadamard state being one that has the same singularity structure of local operator expectation values as the vacuum itself. These form a large class of physically relevant states in QFT.} We note $Q_3(1)=0$ and $\lim_{x\rightarrow\infty}Q_3(x)=1$. If the field is massless then this bound is
\beq
    \int\dd t\,g(t)^2\langle \hat T_{00}(t,\vec x)\rangle_\psi\geq-\frac{1}{16\pi^2}\int\dd t\,\abs{g''(t)}^2~.
\eeq

This bound applied to $g_\lambda(t):=\lambda^{-1/2}g(t/\lambda)$ illustrates the scaling behavior mentioned above:
\beq
    \int\frac{\dd t}{\lambda}g(t/\lambda)^2\langle \hat T_{00}(t,\vec x)\rangle_\psi\geq-\frac{1}{16\pi^2\lambda^4}\int\dd t\abs{g''(t)}^2~.
\eeq
The limit that $\lambda\rightarrow 0$ shrinks the support of $g_\lambda$ to a point and we find the lower bound trivializes. Oppositely, in the limit that $\lambda\rightarrow\infty$, $g_\lambda$ becomes constant and we derive the {\it averaged weak energy condition} (AWEC):
\beq\label{eq:AWEC}
    \int_{-\infty}^\infty\dd t\,\langle\hat T_{00}(t,\vec x)\rangle_\psi\geq 0~.
\eeq

The key physical significance of this QEI that will persist for all QEIs is that while negative energy densities can exist along a geodesic, they are eventually compensated for elsewhere along that geodesic. This was described by Ford and Roman as {\it quantum interest} \cite{Ford:1999qv}: one can borrow a small amount negative energy density for a short time, but this must be `paid back' with a greater amount of positive energy density. The AWEC \eqref{eq:AWEC} indicates that, in the end, the one making the loan comes out ahead.

\subsubsection*{State-independent QEIs}

The worldline bound \eqref{eq:FewsterEvesonBound} is a particular example of a {\it state-independent} QEI in which the lower bound involves only $\mathbb C$-number quantities depending on the smearing function. The general structure of a state-independent QEI is the following:
\beq
    \int_{D}\dd^p x\,g(x)^2\,\langle\hat\rho(x)\rangle_\psi\geq - \mc D^{(\rho)}[f]~,
\eeq
where $\hat \rho$ the operator of interest (such as an energy density or a null-energy density), $D$ is some codimension-$(d-p)$ spacetime region on which $f$ is compactly supported, and $\mc D^{(\rho)}[\cdot]$ is a distributional map on $f$ (i.e. it could be a differential operator or could involve a Fourier transform), whose form can depend on the operator in question. 

This is a somewhat restrictive form of QEI and not all fall into this class. However for most of this lecture I will focus on some notable state-independent QEIs. I will return to talk shortly about state-dependent QEIs and their interpretation at the end of this lecture.

\subsubsection*{Free field construction of state-independent QEIs in Minkowski space}

Let $\rho$ be a classically positive quantity. It should be expressed as a sum (with positive coefficients) of squares of some combinations of the elementary fields and w.l.o.g we will take the case that $\rho$ itself is a perfect square, $\rho(x)=\mc O(x)\mc O(x)$ for some $\mc O$. Upon quantization we renormalize $\rho$ by normal ordering:
\beq
    :\hat\rho(x): =\hat\rho(x)-\langle\hat\rho(x)\rangle_\Omega~.
\eeq
We now consider the smeared quantity
\beq
    \mc Q_\rho[g]=\int_{D_p}\dd^px\,g(x)^2\langle :\hat\rho(x):\rangle_\psi~,
\eeq
where $D_p$ is a $p$-dimensional timelike subspace of $\mathbb R^{1,d-1}$. We can artificially ``point-split" the composite $\mc O\mc O$ along $D_p$ by a subspace Fourier transform:
\beq\label{eq:Qrho1}
    \mc Q_\rho[g]=\int_{\tilde D_p}\frac{\dd^p\xi}{(2\pi)^p}\int_{D_p}\dd^px\,\dd^px'g(x)g(x')e^{i\xi\cdot(x-x')}\left(\langle \mc O(x)\mc O(x')\rangle_\psi-\langle\mc O(x)\mc O(x')\rangle_\Omega\right)~.
\eeq
where $\tilde D_p$ is the dual space to $D_p$. The first term is inherently positive and the source of the negativity comes from the vacuum subtraction, which is actually divergent. However there are a lot of cancellations between the two terms that we can take into account before dropping the positive term in bounding $\mc Q_\rho$ from below.\\
\\
\textbf{Assumption:} the commutator of $\mc O$ with itself is a $\mathbb C$-number,
\beq
    [\mc O(x),\mc O(x')]=c(x,x')\hat{1}~.
\eeq
This is a strong assumption as it needs to apply for all points in the timelike domain, $D_p$, not just for spacelike separated pairs. This assumption is obviously true for free fields. For CFTs this assumption strongly restricts the operators that can appear in the $\mc O\mc O$ OPE and implies that $\mc O$ is a (generalized) free field \cite{Fliss:2021phs}.

This assumption implies that the anti-symmetric part of \eqref{eq:Qrho1} cancels in the vacuum subtraction and so we can replace $\tilde D_p$ by a half-space $\tilde H_p=\{\xi\in\tilde D_p\big|\xi_0\geq 0\}$ of positive frequency:
\beq
    \mc Q_\rho[g]=2\int_{\tilde H_p}\frac{\dd^p\xi}{(2\pi)^p}\int_{D_p}\dd^px\,\dd^px'g(x)g(x')e^{i\xi\cdot(x-x')}\left(\langle \mc O(x)\mc O(x')\rangle_\psi-\langle\mc O(x)\mc O(x')\rangle_\Omega\right)~.
\eeq
The first term is still inherently positive and the second term inherently negative, but under the cancellation they are both {\it finite}. This is because the cancellation has softened the contact divergence, $\int_{\tilde D_p}\frac{\dd^p\xi}{(2\pi)^p}e^{i\xi\cdot(x-x')}=\delta^p(x-x')$ to $\int_{\tilde H_p}\frac{\dd^p\xi}{(2\pi)^p}e^{i\xi(x-x')}\sim\frac{\delta^{(p-1)}(\vec x-\vec x')}{x_0-x_0'}$. Thus we can drop the inherently positive term and arrive a finite lower bound for $\mc Q_\rho$:
\beq\label{eq:freeQEI1}
    \mc Q_\rho[g]\geq-2\int_{\tilde H_p}\frac{\dd^p\xi}{(2\pi)^p}\int_{D_p}\dd^px\dd^px'g(x)g(x')e^{i\xi(x-x')}\langle \mc O(x)\mc O(x')\rangle_\Omega~.
\eeq
This bound is state-independent in that it is determined only by the vacuum of the QFT and applies for the expectation values of all states $\psi$ with the same singularity structure as the vacuum (that is, Hadamard states). Translation invariance of the vacuum additionally implies that in momentum space
\beq
    \langle \tmc O(q)\tmc O(q')\rangle_\Omega=(2\pi)^p\delta^{p}(q+q')\mc G_{\rho}(q)~,
\eeq
for some $\mc G_{\mc O\mc O}$ and so our QEI is most easily expressed in momentum space:
\beq\label{eq:freeQEI2}
    \mc Q_\rho[g]\geq-2\int_{\tilde D_p}\frac{\dd^pq}{(2\pi)^p}\abs{\tilde g(q)}^2\int_{\tilde H_p}\frac{\dd^p\xi}{(2\pi)^p}\mc G_\rho(q-\xi)~.
\eeq
One can check that \eqref{eq:freeQEI2} includes the timelike worldline bound described above, \eqref{eq:FewsterEvesonBound}.

\subsection*{ANEC}

The most celebrated state-independent energy inequality is the {\it averaged null energy condition} which states that the integral of the null-energy density integrated over an entire null geodesic, $\gamma$, is a positive operator:
\beq\label{eq:ANEC}
    \int_{\gamma}\dd u\,k^\mu k^\nu\langle T_{\mu\nu}(u,v;\vec x_\perp)\rangle_\psi\geq 0~,
\eeq
where $k$ is the tangent to $\gamma$. For free fields in Minkowski spacetime the ANEC follows from boosting the AWEC to a null hypersurface. We can also verify this directly in Fock quantization. For notational simplicity let $k=\pa_+:=(\pa_0+\pa_1)$. The null-stress energy tensor is given by
\begin{align}
    T_{++}=&\int_{\vec q_1}\int_{\vec q_2}\left(a^\dagger_{\vec q_1}a_{\vec q_2}e^{-i(q_1-q_2)_\mu x^\mu}-\frac{1}{2}(a^\dagger_{\vec q_1}a^\dagger_{\vec q_2}e^{i(q_1+q_2)_\mu x^\mu}+a_{\vec q_1}a_{\vec q_2}e^{-i(q_1+q_2)_\mu x^\mu})\right)~.
\end{align}
Again the first term, the number conserving one, is inherently positive. In this case, however, the integral $\int\dd x^+$ completely removes the second term (at least for states where it doesn't diverge). This is because it brings down a delta function
\beq
    (2\pi)\delta(q_{1+}+q_{2+})~,
\eeq
which can't be satisfied and so vanishes as a distribution.\footnote{We recall that the positive frequency integration range in canonical quantization is
\beq
    q_+q_-\geq\vec q_{\perp}^2+m^2~,\qquad q_+\geq 0~,
\eeq
when stated in lightcone coordinates.}

Much like the Hamiltonian, the ANE and the AWE operators are allowed to be positive operators because they annihilate the vacuum and they are not identically zero because they are integrated over an entire causal geodesic: there are no spacelike-separated operators on which to run the Reeh-Schlieder argument.

The ANEC is perhaps also the greatest success story of a quantum energy inequality: it has been proven for {\it all} quantum field theories in Minkowski spacetime, including interacting ones \cite{Hartman:2016lgu,Faulkner:2016mzt}. The proofs in these two papers are significantly different and illustrate the intricate relation between energy and causality, and energy and quantum information, respectively. We will return to give a sketch of the latter proof in the following lecture. The ANEC is also practical and can be used to restrict long-traversable wormholes and other topologically pathological spacetimes. I will revisit this point in Lecture \ref{sect:l4}. This utility is in fact two-way. For 4d CFTs the ANEC implies the Hofman-Maldacena bounds on the $a$ and $c$ central charges of the theory \cite{2008_Hofman}:
\beq\label{eq:HMbounds}
    \frac{31}{18}\geq \frac{a}{c}\geq \frac{1}{3}~,
\eeq
with the upper and lower bounds saturated by free vectors and free scalars, respectively. Recently the positivity of the ANEC was used to establish new monotonicity theorems along renormalization group flows \cite{Hartman:2023qdn,Hartman:2023ccw}.

However, the ANEC also has its drawbacks. For one, it has only been proven in generality in Minkowksi space; a curved background proof only pertains to free fields, \cite{Kontou:2015yha}. Moreover its pseudo-global nature (integration over an entire null geodesic) makes it a severe departure from the classical, pointwise, energy conditions and severely restricts its validity and thus utility. As an easy example we can consider a QFT on $\mathbb R_t\times S^1\times X_{d-2}$ and look at the null-energy density averaged over the null geodesic wrapping the $S^1$, as depicted in Figure \ref{fig:AANEC}. Because of the compact direction, the QFT has a negative Casimir energy density which then integrates to negative infinity:
\beq
    \int_\gamma \dd \lambda\,\langle T_{kk}\rangle \sim -\int_\gamma \dd\lambda\,\rho_\text{Cas}=-\infty~.
\eeq
\begin{figure}[h!]
\centering
    \begin{tikzpicture}[isometric view]
\draw[->] (-2,0,1) -- (-2,0,6);
\node[above] at (-2.5,0,3.5) {$t$};
\node[below,blue] at (-.5,0,1.5) {$\gamma$};
\foreach[count=\jj]\j in {dashed,solid}
{
\draw[black,solid] (135:1) arc (135:-45:1) --++ (0,0,6) arc (-45:135:1) -- cycle;
\draw[black,\j] (135:1) + (0,0,6*\jj-6) arc (135:-45:1); 
\foreach\z in {0,3}
     \draw[dashed,blue,thick] plot[domain=0:180,samples=145] ({cos(\x-45)},{sin(\x-45)},\z+\x/120);
}
   \draw[black,solid] (135:1) arc (135:315:1) --++ (0,0,6) arc (315:135:1) -- cycle;
   \foreach\z in {0,3}
     \draw[blue,thick,-Stealth] plot[domain=180:360,samples=145] ({cos(\x-45)},{sin(\x-45)},\z+\x/120);
    \draw[thick,red] (.5,0,1.5) -- (.5,0,4.5);
\end{tikzpicture}
\caption{A null geodesic, $\gamma$, wrapping a compact direction in spacetime. Because of Casimir energy, the averaged null energy along this geodesic is infinitely negative. The causal line, in red, connecting separate points in $\gamma$ exclude this example from the achronal ANEC.}\label{fig:AANEC}
\end{figure}
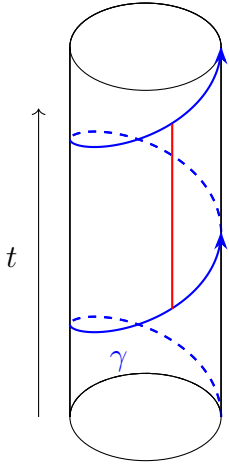
This has led to an improvement of the ANEC known as the {\it achronal ANEC} (AANEC) which states that the ANEC, \eqref{eq:ANEC}, holds for all achronal $\gamma$ (an achronal path being one in which no two separate points can be connected by a causal line segment; see Figure \ref{fig:AANEC}). This eliminates the Casimir example because the geodesic in question is obviously chronal. The condition on chronality is much more restrictive however: not every spacetime admits a complete, achronal, null geodesic. It would be desirable if we had a more local bound on the null energy density; let us explore some attempts in this direction below.

\subsection*{2d CFTs}

A bound on integrated null-energy density was proposed for free field theories in two dimensions by Flanagan in 1997 \cite{Flanagan:1997gn} and was later proven by to hold generically in 2d CFTs by Fewster and Hollands in 2004 \cite{Fewster:2004nj}. The Fewster-Hollands argument is neat and goes as follows. Consider the stress-energy tensor in a 2d CFT. Tracelessness implies $T_{-+}=0$ and so conservation implies that the null stress-tensor is only a function of $x^+$, $\pa_-T_{++}=0$. We integrate it along the null line with a smooth smearing function of compact support,
\beq
    \mc Q_{++}[g]=\int\dd x^+\,g(x^+)^2\langle T_{++}(x^+)\rangle_\psi~.
\eeq
We now perform a coordinate transformation $x^+\rightarrow u(x^+)$. Because of the infinite dimensional symmetry of 2d CFTs there exists a unitary operator, $\hat U$, implementing this coordinate change up to the anomaly given by the Schwarzian derivative, $\{u,x\}=\frac{u'''}{u'}-\frac{3}{2}\left(\frac{u''}{u'}\right)^2$:
\beq
    \hat U^\dagger T_{++}(x^+)\hat U=u'(x^+)^2T_{uu}(u(x^+))-\frac{c}{24\pi}\left\{ u(x^+),x^+\right\}~.
\eeq
with $c$ the central charge. We can choose $u'(x^+)=g(x^+)^{-2}$ to `uniformize' the integral over the null line and find
\beq\label{eq:2dQpptransformed}
    \mc Q_{++}[g]=\int \dd u\langle T_{uu}(u)\rangle_{U^\dagger\psi}-\frac{c}{24\pi}\int \dd x^+\,\left(g'(x^+)\right)^2~.
\eeq
The first term is the ANE quantity in the state $U^\dagger|\psi\rangle$ and is positive. Thus we find the semi-local null-energy bound
\beq\label{eq:2dFlanaganBound}
    \int\dd x^+\,g(x^+)^2\langle T_{++}(x^+)\rangle_\psi\geq -\frac{c}{24\pi}\int\dd x^+g'(x^+)^2~.
\eeq
Moreover, this bound is tight: we know, in principle, the state that saturates it and so there is no bound that can improve on \eqref{eq:2dFlanaganBound}. This is because we have already argued that the ANE is allowed to be positive as an operator because it annihilates the vacuum exactly. Thus if $|\psi\rangle=U|\Omega\rangle$, the ANE portion of \eqref{eq:2dQpptransformed} exactly vanishes which saturates \eqref{eq:2dFlanaganBound}.

\subsection*{SNEC}

Can we promote the integrated null-bound of Flanagan to higher dimensions? We gave some arguments above that smeared equalities over compact spacelike regions are generally unbounded while we know that a compact timelike wordline smearing is lower bounded (at least in free field theory), e.g. \eqref{eq:FewsterEvesonBound}. It is unclear where a finite null segment show sit between these two classifications. However it was shown by Fewster and Roman \cite{Fewster:2002ne} that the null stress-energy tensor integrated along a null segment diverges to $-\infty$ amongst the class the ``0+2" states in 4d free scalar field theory. The origin of this divergence is clear from dimensional analysis: in 2d the stress tensor has dimension $\Delta=2$ and boost-weight two (i.e. two lower ${\scriptstyle (+)}$ indices). Thus it can be lower bounded by a quantity with two null derivatives $\sim \pa_+^2$ and no other length scales. However in higher dimensions $T_{++}$ has dimension $\Delta=d$ and boost-weight two. Thus we can soak up the boosts with two null derivatives, however it is unclear (in flat space) what length scale should soak up the other $d-2$ dimensions in the bound. As is wont to happen in QFTs, any opening for a divergence leads to a divergence and this length scale is the UV cutoff.

Freivogel and Krommydas conjectured an alternative, proposing that in a QFT consistently coupled to dynamical gravity the natural cutoff is the Planck-scale. This led to the conjectured {\it smeared null energy condition} (SNEC) \cite{Freivogel:2018gxj}:
\beq\label{eq:gravSNEC}
    \int_{-\infty}^\infty \dd x^+g(x)^2\langle T_{++}(x)\rangle_\psi\geq-\frac{1}{16\pi G_N}\int_{-\infty}^\infty\dd x^+g'(x^+)^2~.
\eeq
The SNEC is an interesting (and powerful) bound because it involves both QFT and gravity variables. Later we will see how the SNEC can be used as an ingredient for an improved singularity theorem even in the presence of initial NEC violation. However this interplay of QFT and gravity makes gathering evidence for the SNEC difficult and it remains conjectural. Leichenauer and Levine \cite{Leichenauer:2018tnq} showed that it holds for holographic CFTs coupled to dynamical gravity through braneworld holography and using the principle of `no bulk shortcut' (Ref. \cite{Leichenauer:2018tnq} also fixed the precise coefficient which does not follow from dimensional analysis). However evidence for the SNEC directly from QFT coupled to dynamical gravity (or even QFT on a fixed curved background) remains unavailable.

Fliss and Freivogel \cite{Fliss:2021gdz} used the properties of free fields quantized on a light-sheet, to construct a field-theory SNEC for free field theories equipped with a UV cutoff. Because the null quantization of free field theory will come in handy when we discuss QNEC below let us summarize some of the key points of that construction.

The salient point is the Hilbert space of free fields quantized on a lightsheet, $\mc L$, becomes {\it ultralocal} in the transverse dimension \cite{Burkardt:1995ct,Wall:2011hj}. The vacuum and the operator algebra factorize into algebras lie along on each null ray. In practice to make sense of this we need to ``coarse grain'' them into small ``pencils''
\beq
    \mc L=\bigcup_{\mathfrak p}\mc P_{\mathfrak p}~,\qquad |\Omega\rangle_{\mc L}=\bigotimes_{\mathfrak p}|\Omega_{1+1}\rangle_{\mathfrak p}~,
\eeq
of transverse width $a^{d-2}$ as depicted in Figure \ref{fig:pencils}; this coarse graining implies that $a^{-1}$ plays the role of a UV cutoff on the transverse momenta \cite{Fliss:2021gdz}. 
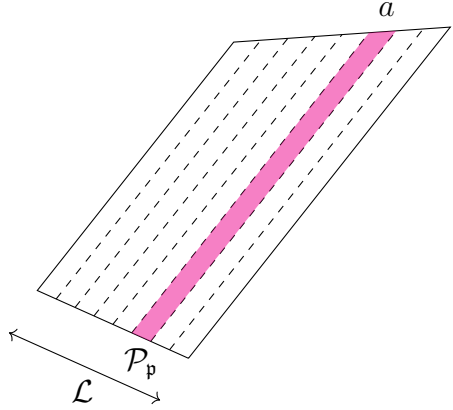
\begin{figure}[h!]
\centering
\begin{tikzpicture}[3d view,perspective]
    \draw (0,0,0) -- (0,4,0) -- (3,4,3) -- (4,0,4) -- (0,0,0);
    \foreach \j in {.5,1,1.5,2,2.5,3,3.5} {
     \draw[dashed] (0,\j,0) -- (4-.25*\j,\j,4-.25*\j);
        }
    \fill[magenta,opacity=.5] (0,1,0) -- (0,1.5,0) -- (3.625,1.5,3.625) -- (3.75,1,3.75) -- (0,1,0);
    \node[above] at (3.75,1.25,3.75) {$a$};
        \draw[<->] (0,.75,-.75) -- (0,4.75,-.75);
    \node[below] at (0,2.8,-.8) {$\mathcal{L}$};
    \node[below] at (0,1.25,0) {$\mathcal P_{\mathfrak{p}}$};
\end{tikzpicture}
\caption{A lightsheet $\mathcal L$ coarse grained into disjoint collections of null rays, i.e. pencils, of transverse width $a^{d-2}$.}\label{fig:pencils}
\end{figure}
The operator algebra of the free field commutes amongst pencils and on each pencil is identical to that of a 2d chiral CFT:
\beq
    [\phi(x^+,\vec y_{\perp,\mathfrak p}),\pa_+\phi(x^{+'},\vec y'_{\perp,\mathfrak p'})]=\frac{i}{a^{d-2}}\delta_{\mathfrak p\mathfrak p'}\delta(x^+-x^{+'})~,
\eeq
where the lightsheet field $\hat \phi$ is related to a chiral CFT field on the pencil $\mc P_\mathfrak p$ by $\hat \phi(x^+,\vec y_{\perp,\mathfrak p})=a^{\frac{2-d}{2}}\varphi_{\mathfrak p}(x^+)$. For the null energy density integrated along the null line contained in pencil $\mathcal P_{\mathfrak p}$ we can utilize the properties of the chiral CFT living on that pencil and the 2d bound \eqref{eq:2dFlanaganBound} to find
\beq\label{eq:fieldSNEC}
    \int\dd x^+g(x^+)^2\langle T_{++}(x)\rangle_\psi\geq-\frac{1}{24\pi a^{d-2}}\int\dd x^+g'(x^+)^2~
\eeq
which we call the ``free field SNEC."

\subsection*{DSNEC}

It would be nice to have a `regulated' SNEC directly accessible in QFT without needing to specify the UV cutoff, or without the special properties of free fields on a lightsheet. The key alluded above is the need to specify another length scale. We also know that the structure of the free QEIs, \eqref{eq:freeQEI2}, works when the smearing is over a subspace, $D_p$, with timelike extent. This leads us to consider smearing the null energy density along {\it two} null directions, $x^\pm$. This provides an invariant length $\ell^2\sim-\delta^+\delta^-$, where $\delta^\pm$ are the smearing length scales in the $x^\pm$ direction, respectively. The result is the {\it double smeared null energy condition} (DSNEC) \cite{Fliss:2021gdz,Fliss:2021phs} which takes the schematic form
\beq\label{eq:DSNEC}
    \int \dd^2x^\pm g(x^+,x^-)^2\langle T_{++}(x)\rangle_\psi\geq-\frac{\mc N[\mu]}{(\delta^+)^{\frac{d+2}{2}}(\delta^-)^{\frac{d-2}{2}}}~,
\eeq
where $\mc N[\mu]$ is a function of any mass scales measured in units of the smearing length, $\mu^2=\delta^+\delta^-m^2$, and scaling linearly with the number of fields. For free fields the DSNEC follows from the general procedure we outlined above leading to \eqref{eq:freeQEI2}, however since it does not involve any theory specific UV regulator, it has hope of being true more broadly. Additionally by taking the smearing region to be long and thin, $(\delta^+,\delta^-)\rightarrow (\infty,0)$ holding $\ell^2\sim\delta^+\delta^-$ fixed we recover the ANEC. Alternatively, in the limit that $\delta^-\rightarrow 0$ with $\delta^+$ fixed (which then implies that $\ell^2\sim\delta^+\delta^-\rightarrow0$) we arrive at a geometrically regulated SNEC,
\beq
    \int\dd x^+ g(x^+)^2\langle T_{++}\rangle \geq -\frac{\mc N}{\ell^{d-2}(\delta^+)^2}~.
\eeq

\subsection*{State-dependent QEIs}

Having covered many nice state-independent QEIs, in this final section we point out that sometimes life and quantum field theory are not so nice to us. Sometimes there just is no single $\mathbb C$-number lower bounding an energy operator regardless of how much we smear it.

For free fields there is a one-to-one correspondence: a quadratic field operator, $\hat\rho$, admits a state-independent lower bound if and only if the classical quantity, $\rho$, is pointwise positive. The argument being that for classically non-positive $\rho$ there are one-particle states with negative expectation value (i.e. even the number conserving portion of $\hat\rho$ is non-positive). However since the theory is free, nothing stops us from building multi-particle states of arbitrary particle number by piling this one-particle state on top of itself. In this way one can build a state of arbitrarily negative $\langle\hat\rho\rangle$ regardless of how it is smeared.

We have already seen an example of classically non-positive energy operators in free scalar field theory in Lecture \ref{sect:l1}. E.g. the non-minimally coupled scalar field has a non-positive null stress tensor, $T_{++}$, \eqref{eq:NMCTkk}. As expected one can find explicit one-particle states realizing negative energy density \cite{Fewster:2007ec} and null energy density \cite{Fliss:2023rzi}, from which a series of multi-particle states can be constructed pushing those respectively quantities arbitrarily negative.

Not all hope is lost: states with large negative energy density are states are large particle number which mean they are also states of large field expectation value $\langle :\phi^2:\rangle$. This suggests that the lower bound on the energy density could depend on the field value expectation value. Such a {\it state-dependent QEI} for the non-minimally coupled scalar was first constructed by Fewster and Osterbrink in considering the energy density. Here we will state the corresponding quantity for the DSNEC constructed in \cite{Fliss:2023rzi}
\beq\label{eq:NMCDSNEC}
    \int\dd^2x^\pm g(x^\pm)^2\langle T_{++}(x)\rangle_\psi\geq-\text{DSNEC}[g]-\xi\int\dd^2x^\pm \langle :\phi(x)^2:\rangle_\psi\pa_+^2\left(g(x^\pm)\right)^2
\eeq

This suggests a ``next-best'' structure for state-dependent QEIs which are {\it linear}, i.e. they depend on the state only through the expectation values of operators:
\beq\label{eq:statedepQEIgen}
    \int_{D_p}\dd^px\,g(x)^2\langle\hat\rho(x)\rangle_\psi\geq-\mc D_\rho^{(1)}[g]-\sum_{\mc O}\int_{D_p}\dd^px\,\mc D^{(\mc O)}_\rho[g]\langle \mc O(x)\rangle_\psi~.
\eeq
A state-dependent QEI like \eqref{eq:statedepQEIgen}works best when the operators appearing on the right-hand side have a softer divergence structure in Hadamard states than the stress tensor itself (say, because they have lower scaling dimension). It is hard to say if the necessity of state-independent QEIs is an accident of free theory and the ability to pile states of arbitrary particle number (there is indication in large-$N$ interacting CFTs that this procedure fails and a state-independent bound may be possible \cite{Fliss:2024dxe}) or is too na\"ive and the most general bound is non-linear in the state (we will see such a bound below).

If it is the case that state-dependent QEIs are the general structure of interacting theories, one optimistic scenario is if the sum over operators terminates to some finite set. This scenario remains unproved in general, but one might have hope of proving such a bound in CFT using properties of the OPE:
\beq\label{eq:statedepQEICFThope}
    \int_{D_p}\dd^px\,g(x)^2\langle\hat\rho(x)\rangle_\psi\overset{?}{\geq}-\mc D_\rho^{(1)}[g]-\sum_{\Delta\leq\Delta_\ast}\int_{D_p}\dd^px\,\mc D^{(\mc O)}_\rho[g]\langle \mc O_\Delta(x)\rangle_\psi~,
\eeq
where we've organized the sum in terms of conformal dimensions up to some maximum dimension. In this situation, the finite number of expectation values $\langle \mc O_\Delta\rangle_\psi$ provide a set of parameters defining an effective field theory in which energy densities are well-behaved. Such a perspective was advocated for the non-minimally coupled scalar field in \cite{Fliss:2023rzi} where it was argued that semi-classical gravity mimimally coupled to a scalar field is a sensible effective field theory for states of small $\langle :\phi^2:\rangle$.

\section{Lecture 4: Constraining semi-classical gravity}\label{sect:l4}

At this point in the lectures we are now experienced carpenters, having constructed many structures of the `low-grade wood' side of \ref{eq:marblewood}. At this point we can utilize some of what we constructed to provide a base for the `fine marble' side of semi-classical gravity.

\subsection*{ANEC and topological censorship}

As a first example, let us examine the {\it topological censorship conjecture} \cite{Friedman:1993ty}. This is sometimes known more colloquially as ``no traversable wormholes'' and in words states that no causal observer should be able to probe topologically non-trivial features of in an asymptotically flat spacetime. More specifically, we can imagine a spacetime $M$ with two asymptotically flat regions connected by a spacelike wormhole as in the maximally extended Schwarzschild spacetime, depicted in Figure \ref{fig:maxSchwarz}. In the situation in which a spacetime has a wormhole connecting two separate regions of the same asymptotically flat region we can pass to the first situation by considering the double cover without affecting the following arguments. 

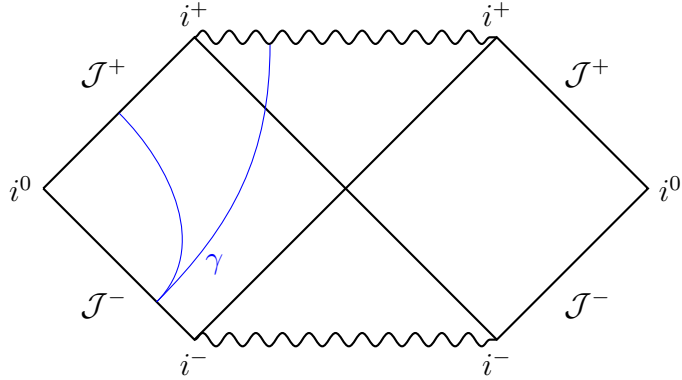
\begin{figure}[h!]
\centering
\begin{tikzpicture}
    \draw[thick] (0,0) -- (2,2) -- (6,-2) -- (8,0) -- (6,2) -- (2,-2) -- (0,0);
    \draw[thick,style={decorate,decoration=snake}] (2,2) -- (6,2);
    \draw[thick,style={decorate,decoration=snake}] (2,-2) -- (6,-2);
    \draw[blue,smooth] (1.5,-1.5) to[out=45,in=-45] (1,1);
    \draw[blue,smooth] (1.5,-1.5) to[out=45,in=-90] (3,1.9);
    \node[above] at (6,2) {$i^+$};
    \node[above] at (2,2) {$i^+$};
    \node[below] at (6,-2) {$i^-$};
    \node[below] at (2,-2) {$i^-$};
    \node[above] at (.8,1.2) {$\mathcal J^{+}$};
    \node[below] at (.8,-1.2) {$\mathcal J^{-}$};
    \node[above] at (7.2,1.2) {$\mathcal J^{+}$};
    \node[below] at (7.2,-1.2) {$\mathcal J^{-}$};
    \node[left] at (0,0) {$i^0$};
    \node[right] at (8,0) {$i^0$};
    \node[blue,right] at (2,-1) {$\gamma$};
\end{tikzpicture}
\caption{The Penrose diagram of the maximally extended Schwarzschild spacetime. A causal curve, $\gamma$, emanating from past null infinity, $\mathcal J^-$, can either escape to future null infinity $\mathcal J^+$ within the same connected component, or it can fall into the singularity. It cannot however traverse the spacelike wormhole to reach a separate component of null infinity.}
\label{fig:maxSchwarz}
\end{figure}

\begin{theorem}[Topological censorship]
    Consider an asymptotically flat, globally hyperbolic spacetime, $M$ that satisfies the ANEC. Then any causal curve from past null infinity, $\mc J^-$, to future null infinity, $\mc J^+$, is deformable to a timelike curve from $\mc J^-$ to $\mc J^+$ in a single connected component of $\mc J^-\cup\mc J^+$.
\end{theorem}
The last sentence is simply telling us that causal curves that emerge and escape to asymptotic infinity do so entirely within a single asympotic region; see Figure \ref{fig:TopCens}.

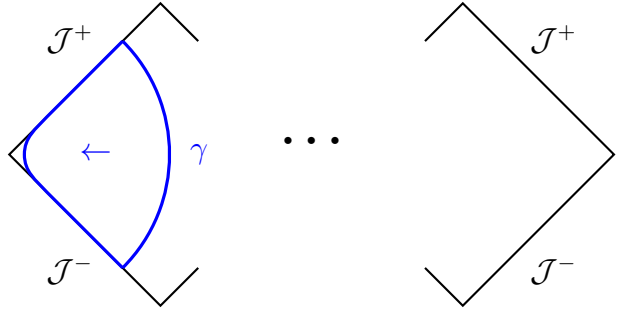
\begin{figure}[h!]
    \centering
    \begin{tikzpicture}
    \draw[thick] (2.5,-1.5) -- (2,-2) -- (0,0) -- (2,2) -- (2.5,1.5);
    \draw[thick] (5.5,1.5) -- (6,2) -- (8,0) -- (6,-2) -- (5.5,-1.5);
    \draw[blue,smooth,very thick] (1.5,-1.5) to[out=45,in=-45] (1.5,1.5);
    \draw[blue,smooth,very thick] (1.5,-1.5) to[out=135,in=-45] (.5,-.5) to[out=133,in=-90] (.2,0) to[out=90,in=-133] (.5,.5) to[out=45,in=-135] (1.5,1.5);
    \node[above] at (.8,1.2) {$\mathcal J^{+}$};
    \node[below] at (.8,-1.2) {$\mathcal J^{-}$};
    \node[above] at (7.2,1.2) {$\mathcal J^{+}$};
    \node[below] at (7.2,-1.2) {$\mathcal J^{-}$};
    \node[blue,right] at (2.25,0) {$\gamma$};
    \node[blue,left] at (1.5,0) {$\leftarrow$};
    \node[above] at (4,0) {$\Large{\boldsymbol{\ldots}}$};
\end{tikzpicture}
\caption{Topological censoreship: regardless of what happens in the interior of the spacetime, a causal curve $\gamma$ from $\mathcal J^-$ to $\mathcal J^+$ must be deformable to a timelike curve connecting a single connected component of $\mathcal J^-\cup\mathcal J^+$.}\label{fig:TopCens}
\end{figure}

Even though the theorem only requires the ANEC, we will start with the classical argument, assuming the NEC. Consider a null geodesic on the left asymptotic region shot outward and future directed. Very far away from the wormhole the metric is locally flat and takes the form
\beq
    \dd s^2\approx -\dd t^2+\dd r^2+r^2\dd\Omega_{d-2}^2~.
\eeq
Null geodesics shot inward (i.e. to decreasing $r$) will have negative expansion. As we by now are well aware, the null Raychaudhuri equation, paired with the null energy condition implies that this null congruence develops a focal point in finite affine time.

For the purposes of topological censorship however, we do not need the pointwise NEC but instead the much weaker ANEC. We will not fully reproduce the details here (although see \cite{Witten:2019qhl} for a good summary) however physically we can reason this as the following. A causal curve emanating from $\mc J^-$ to $\mc J^+$ will be extentable to infinite affine parameter. Thus even if there is a small region of NEC violation, there is a lot of affine time for positive null energy to win over and develop focal points.

More specifically, we can consider the notion of {\it promptness}. A prompt causal curve between two points is one that arrives ``as quick as possible.'' To be more precise:\\
\\
{\it A causal curve $\gamma$ from $q$ to $p$ is prompt if there does not exist a causal curve $\gamma'$ from $q$ to $p'$ with $p'$ in the past lightcone of $p$.}\\
\\
Prompt curves are null geodesics without any focal points shared with nearby null geodesics. If a two separate asymptotic regions of $M$ are causally connected then there exists a causal curve from a component $\mc J^-$ in one region to a component of $\mc J^+$ in another region, and so there also is a prompt causal curve obeying this. 

We consider the conditions for causal curve to be prompt along an interval by looking small deformations maintaining the causality. To see what this implies, we will pick Fermi normal coordinates $\{u,v, x^A\}_{A=1,\ldots, d-2}$ along $\gamma$, such that in a neighborhood of $\gamma$ the metric takes the form
\beq
    \dd s^2=-2\dd u\dd v+\sum_{A=1}^{d-2}(\dd x^A)^2+O(v^2,vx^A,x^Ax^B)~,
\eeq
and the geodesic equation, $\frac{D x^\mu}{D\lambda^2}=0$ is satisfied for the parameterization $u=\lambda$ and $v=x^A=0$. We also have along $\gamma$
\beq
    R_{AuBu}=-\frac{1}{2}\pa_A\pa_B\,g_{uu}~.
\eeq
In these coordinates let the initial point of $\gamma$, be at $(u,v,x^A)=(\lambda_0,0,0)$ and the endpoint at $(u,v,x^A)=(\lambda_1,0,0)$. We will consider a causal deformation of our null geodesic, $x^{\mu}(\lambda)=(\lambda,0,0)\rightarrow (\lambda,\epsilon^2\, v(\lambda),\epsilon x^A(\lambda))+O(\epsilon^3)$. The requirement of a causal curve is that
\beq
    g_{\mu\nu}(x(\lambda))\frac{\dd x^\mu}{\dd\lambda}\frac{\dd x^\nu}{\dd\lambda}\leq 0~,\qquad\qquad \forall \lambda~.
\eeq
At order of $\epsilon^2$, causal curves then satisfy
\beq
    2\frac{\dd v}{\dd\lambda}-\left(\frac{\dd x^A}{\dd\lambda}\right)^2+x^Ax^B\,R_{AuBu}\geq 0~.
\eeq
Integrating this equation we have
\beq\label{eq:vgeqJ}
    v(\lambda_1)\geq \frac{1}{2}J[x^A]~,
\eeq
with
\beq
    J[x^A]=\int_{\lambda_0}^{\lambda_1}\dd\lambda\left[\left(\frac{\dd x^A}{\dd\lambda}\right)^2-x^Ax^B\,R_{AuBu}\right]~.
\eeq
Recalling the definition of promptness, if we can find a single deformation with $x^A(\lambda_0)=x^B(\lambda_1)=0$ such that $J[x^A]$ is negative, then we can let $v(\lambda)$ saturate \eqref{eq:vgeqJ} and have $v(\lambda_1)<0$ putting it in the past light cone of the endpoint of $\gamma$. $\gamma$ is prompt only if no such situation arises. Said another way,
\beq
    \gamma\text{ is prompt }~\Rightarrow~ J[x_1^A]\geq 0~\forall~ x^A(\lambda)\text{ with }x^A(\lambda_0)=x^A(\lambda_1)=0~.
\eeq
This needs to be true for entire infinite null line extending between asymptotic regions. By taking $x^A$ to be a constant and averaging over those constants\footnote{If one is concerned about the boundary conditions $x^A(-\infty)=x^A(\infty)=0$, one can take large Gaussians in the limit of infinite width.}, one then finds this can only be if
\beq
    -\int_{-\infty}^\infty\dd\lambda R_{uu}(x(\lambda))\geq 0~.
\eeq
The ANEC then implies that there is no such prompt causal curve.

\subsection*{QEIs and improved singularity theorems}

The above proof for topological censorship falls under a general method of proving singularity theorems called `index methods,' introduced by O'Neill \cite{ONeill:1983semi}. There $J$ is the index of the null curve, $\gamma$. The key methodology is look at the second variation of the proper length functional away from either a null or timelike geodesic and seeing if resulting variational operator is positive or negative definite. In the previous section on topological censorship, we saw that existence of a focal point corresponded to a negative index when expanded about a null geodesic. The key benefit of index methods is they naturally implement integrated conditions as criteria for development of focal points within a given region. This makes them a natural starting point for utilizing QEIs in singularity theorems.

Following O'Neill's methods, Brown, Fewster and Kontou \cite{Brown:2018hym} (and subsequently Fewster and Kontou \cite{Fewster:2019bjg,Fewster:2021mmz}) developed a set of improved singularity theorems using as input the structure of state-independent QEIs.\footnote{Strictly speaking \cite{Brown:2018hym} establishes a singularity theorem using the classical SEC however uses as input an energy condition expressed in the same form as a QEI.} To be specific, \cite{Brown:2018hym,Fewster:2019bjg,Fewster:2021mmz} use as input for their singularity theorems the following {\it worldline curvature condition}. For any timelike geodesic, $\gamma$,
\beq\label{eq:WLcurvcond}
    \int_\gamma \dd\tau g(\tau)^2\,t^{\mu}t^{\nu}R_{\mu\nu}(\gamma(\tau))\geq -\sum_{\ell=0}^{m}\mc Q_\ell\norm{g^{(\ell)}}~,\qquad \mc Q_\ell\geq 0
\eeq
for some fixed, maximum $m$ and where $\norm{\cdot}:=\int\dd\tau\abs{\cdot}^2$ is the $L^2$-norm on $\gamma$ and the superscript indicates the $\ell$'th derivative. We will now schematically follow the two singularity theorems introduced in \cite{Fewster:2019bjg} which cover the scenarios in which timelike geodesics are initially converging or initially diverging, respectively.

\begin{theorem}[Improved Hawking Singularity Theorem I]\label{th:FK1}
Let $M$ be a globally hyperbolic spacetime with a compact Cauchy slice, $\Sigma$. Suppose that there exists a $\bar\tau>0$ such that for any $g(\tau)$ with compact support in $[0,\bar\tau]$ the worldline curvature condition \eqref{eq:WLcurvcond} holds for some $\{\mc Q_m\}$ and $m\geq 1$. Further suppose there exists a $\tau_0\in[0,\bar\tau)$ and $\rho_0>0$ such that $t^\mu t^\nu R_{\mu\nu}\geq\rho_0$ for all $\tau\in[0,\tau_0]$ and extrinsic curvature of $\Sigma$ is bounded below by a positive constant, $\bar\kappa$, that depends on $\mc Q_m$ (see \cite{Fewster:2019bjg} for the precise relation). Then no future timelike geodesic emanating from $\Sigma$ has proper length greater than $\bar\tau$ and $M$ is future timelike geodesically incomplete.
\end{theorem}

\begin{theorem}[Improved Hawking Singularity Theorem II]\label{th:FK2}
Let $M$ be a globally hyperbolic spacetime with a compact Cauchy slice, $\Sigma$. Suppose that there exists a $\bar\tau,\tau_0>$ such that for any $g(\tau)$ with compact support in $(-\tau_0,\bar\tau)$ the worldline curvature condition \eqref{eq:WLcurvcond} holds from some $\{\mc Q_m\}$ and $m\geq 1$. Further suppose there exists a $\rho_0>0$ such that for all $\tau\in(-\tau_0,0]$, $t^\mu t^\nu R_{\mu\nu}(\gamma(\tau))\leq -\rho_0$ and that the extrinsic curvature of $\Sigma$ is bounded below by a positive constant, $\bar\kappa'$ that depends on $\mc Q_m$ (a different constant than the previous theorem \cite{Fewster:2019bjg}). Then no future timelike geodesic emanating from $\Sigma$ has proper length greater than $\bar\tau$ and $M$ is future timelike geodesically incomplete.
\end{theorem}

We will not prove the above theorems. However they follow from considering the index of a timelike geodesic extending from $[0,\bar\tau]$ (in the first theorem) or $[-\tau_0,\bar\tau]$ in the second theorem, much in the same way we condsidered the index of a null geodesic in the proof of topological censorship. If the worldline curvature condition, \eqref{eq:WLcurvcond}, for any smooth function $g(\tau)$, then given a sufficiently positive extrinsic curvature, it is possible to find a test function such that the index is negative and a focal point is established. Note the curious set-up of the second theorem, Theorem \ref{th:FK2}: we have assumed that worldline curvature is sufficiently {\it negative} before reaching the Cauchy surface $\Sigma$. Fitting with the idea of quantum interest established in Lecture \ref{sect:l3} this means the curvature must then become even more positive into the future of $\Sigma$ (i.e. it costs more to borrow early and pay back late).

A sticking point in the above two theorems that it requires a worldline inequality on effective energy density, $\rho_\text{eff}=t^\mu t^\nu T_{\mu\nu}+\frac{1}{d-2}T$, which as we saw in Lecture \ref{sect:l1} is not classically positive for a scalar field with a potential (or even a mass). Thus given our discussion on state-independent vs. state-independent QEIs in the previous lecture, there does not exist a state-independent QEI of the form \eqref{eq:WLcurvcond} for the massive scalar field. There does exist a state-dependent QEI for the quantum effective energy density $:\hat\rho_\text{eff}:$ for the massive scalar field of the form (assuming $d$ is even) \cite{Fewster:2021mmz}
\beq\label{eq:scalarSECQEI}
    \int_\gamma \dd\tau\,g(\tau)^2\,\langle:\hat\rho_\text{eff}(\gamma(\tau)):\rangle_\psi\geq - \frac{\pi V_{S_{d-2}}}{d(2\pi)^d}\int_\gamma \dd\tau\abs{g^{(d/2)}}^2-\frac{m^2}{d-2}\int_\gamma\dd\tau\,g(\tau)^2\langle :\phi^2(\gamma(\tau)):\rangle_\psi~,
\eeq
where $V_{S_{d-2}}$ is the volume of the $(d-2)$-sphere. Within the effective field theory of states with maximum field values,
\beq
    \langle:\phi^2:\rangle\leq\phi_\text{max}^2~,
\eeq
we can replace \eqref{eq:scalarSECQEI} with
\beq\label{eq:scalarSECQEI2}
    \int_\gamma \dd\tau\,g(\tau)^2\,\langle:\hat\rho_\text{eff}(\gamma(\tau)):\rangle_\psi\geq - \frac{\pi V_{S_{d-2}}}{d(2\pi)^d}\norm{g^{(d/2)}}-\frac{m^2}{d-2}\phi_\text{max}^2\norm{g}~,
\eeq
which indeed implies the worldline curvature condition, \eqref{eq:WLcurvcond}. Using estimates of the early universe Hubble constant derived from the PLANCK experiment, and estimated masses of pions and baryons, \cite{Fewster:2021mmz} illustrate that $\Lambda_\text{CDM}$ model is consistent with the initial conditions implying past geodesic incompleteness within a proper time in the same order of magnitude of current estimated cosmological age of the universe.

We also inquire about the status of an improved Penrose Singularity Theorem. The index methods establishing a focal point require integrating along a null geodesic segment. The natural starting place for the improved Penrose Singularity Theorem is an {\it integrated null curvature condition} \cite{Fewster:2019bjg}, i.e. for any null geodesic $\gamma$,
\beq\label{eq:IntNullcurvcond}
    \int_\gamma\dd\lambda g(\lambda)^2\,k^\mu k^\nu R_{\mu\nu}(\gamma(\lambda))\geq -\sum_{\ell=0}^m\mc Q_\ell\norm{g^{(\ell)}}~,\qquad \mc Q_\ell\geq 0~,
\eeq
for some fixed maximum $m$. Let us first use this as input for singularity theorems and then inquire if QFT is compatible with it as an assumption.

The first of the singularity theorems assumes a small period of initial NEC satisfaction but then allows the NEC to be violated in a time scale well before (putative) singularity formation:

\begin{theorem}[Improved Penrose Singularity Theorem I]\label{th:FK3}
Let $M$ be a globally hyperbolic spacetime with a non-compact Cauchy slices and possessing a codimension-2 trapped surface, $\Gamma$. Furthermore suppose that there exists an $\bar\lambda$ such that for each null geodesic emanating from $\Gamma$ the integrated null curvature condition, \eqref{eq:IntNullcurvcond}, holds for any $g(\lambda)$ compactly supported in $[0,\bar\lambda]$ and for some $m>1$. Furthermore suppose that mean normal curvature of $\Gamma$ is bounded below by some constant $\bar h$ set by the $\{\mc Q_m\}$ appearing in \eqref{eq:IntNullcurvcond} (see \cite{Fewster:2019bjg} for precise relations) and there exists a $\lambda_0\in[0,\bar\lambda)$ such that the null curvature condition $k^\mu k^\nu R_{\mu\nu}$ holds for all future pointing null geodesics emanating from $\Gamma$ for affine times $\lambda\in[0,\lambda_0)$. Then there is a future directed inextendible null geodesic emanating from $\Gamma$ with affine length less than $\bar\lambda$ and $M$ is future null geodesically incomplete.
\end{theorem}

The second theorem removes the assumption of short NEC satisfaction to the future of $\Gamma$ for a short period of NEC violation to the past of $\Gamma$. Much like the second Improved Hawking Singularity Theorem, Theorem \ref{th:FK2}, this scenario works through quantum interest accrued in the integrated null curvature condition \eqref{eq:IntNullcurvcond}: this past NEC violation must be more than compensated for in the future of $\Gamma$.

\begin{theorem}[Improved Penrose Singularity Theorem II]\label{th:FK4}
Let $M$ be a globally hyperbolic spacetime with a non-compact Cauchy slices and possessing a codimension-2 trapped surface, $\Gamma$. Furthermore suppose that there exists an $\bar\lambda$ such that for each null geodesic emanating from $\Gamma$ the integrated null curvature condition, \eqref{eq:IntNullcurvcond}, holds for any $g(\lambda)$ compactly supported in $[0,\bar\lambda]$ and for some $m>1$. Furthermore suppose that mean normal curvature of $\Gamma$ is bounded below by some constant $\bar h'$ set by the $\{\mc Q_m\}$ appearing in \eqref{eq:IntNullcurvcond} (note that $\bar h'$ is different than the $\bar h$ of the previous theorem; see \cite{Fewster:2019bjg} for precise relations) and there exists a $\lambda_0>0$ and $\rho_0>0$ such that the null curvature condition is violated $k^\mu k^\nu R_{\mu\nu}<-\rho_0$ holds for all future pointing null geodesics an intersecting $\Gamma$ and continued to affine times $\lambda\in(-\lambda_0,0]$ to the past of $\Gamma$. Then there exists a future directed inextendible null geodesic emanating from $\Gamma$ with affine length less than $\bar\lambda$ and $M$ is future null geodesically incomplete.
\end{theorem}

The proofs of these theorems work in a very similar manner to those for the improved Hawking Singularity Theorems (Theorems \ref{th:FK1} and \ref{th:FK2}). However the status of the matter that sources them is much murkier. As we argued in the previous lecture, there are no null QEIs (even state-dependent ones) integrated along finite null segments strictly within QFT without introducing a physically motivated cutoff. The closest candidate is the SNEC \eqref{eq:gravSNEC} which was indeed shown to be sufficient input for the Improved Penrose Singularity Theorems \cite{Freivogel:2020hiz}. The authors of that paper used the SNEC, for instance, to show that the NEC violation incurred by Hawking radiation outside the horizon does not spoil the robustness of the singularity theorem. However within the realm of QFT alone, the SNEC remains a conjectured condition. There is hope of utilizing DSNEC as a `regulated' form the SNEC to imply a regulated integrated curvature condition that could feed into a singularity theorem, however to date this has not been achieved.

\section{Lecture 5: Energy and quantum information}\label{sect:l5}

In this lecture we shift gears and focus on a completely new toolkit for discussing energy in QFT. This toolkit lies in the connections to quantum information theory. We have already seen a hint of this connection: the Reeh-Schlieder theorem, which guarantees the non-existence of a positive compactly smeared operator, is directly related to the dense entanglement structure of the quantum vacuum state.

\subsection*{Lightning review of entanglement and entanglement entropy in QFT}

To begin let us recall some basic definitions. Quantum entanglement is a form of information sharing between subsystems of state that is classically forbidden. The quintessential example of an entangled state is the Bell pair of a tensor product of a two-state system:
\beq\label{eq:Bell}
    |\psi_\text{Bell}\rangle=\frac{1}{\sqrt{2}}\left(|\uparrow\rangle_A|\downarrow\rangle_B+|\downarrow\rangle_A|\uparrow\rangle_B\right)~.
\eeq
The systems $A$(lice) and $B$(ob) can be spacelike separated yet a projective measurement by Alice (say yielding an $|\uparrow\rangle_A$) instantaneously collapses Bob's state vector (in this case to $|\downarrow\rangle_B$). Thus there is a non-classical correlation between the systems. Contrapositively, without access to both subsystems, there is a loss of information and the {\it entanglement entropy} is measure of that loss. We define it as the following.

Let $\rho$ be a density matrix\footnote{The notation of $\rho$ for a density matrix is standard in the quantum entanglement literuature. We caution the reader not to equivocate this will all the other instances of the symbol $\rho$ appearing these lectures.} on a tensor Hilbert space admitting a tensor product on two subsystems, $\mc H=\mc H_A\otimes\mc H_B$ (one can have in mind $\rho=|\psi\rangle\langle \psi|$ for some state $|\psi\rangle\in\mc H$). We define the reduced density matrix, $\rho_A$, as the partial trace of $\rho$ over $\mc H_B$:
\beq
    \rho_A=\Tr_{\mc H_B}\rho~.
\eeq
In essence, $\rho_A$ is the state after `forgetting about' subsystem $B$. Alternatively $\rho_A$ is the unique density matrix contained entirely on the algebra of operators acting on $\mc H_A$ reproducing the same measurements as $\rho$ itself of operators on $\mc H_A$. As mentioned above, forgetting $B$ means a loss of information for entangled states and so $\rho_A$ might be a probabilistic mix after the partial trace. The entanglement entropy is the von Neumann entropy of this mixed state:
\beq
    S_A[\rho]=-\Tr_{\mc H_A}\left(\rho_A\log\rho_A\right)~.
\eeq
In QFT we are going to be interested the amount of quantum correlation shared between local subsystems and so we will imagine the scenario when $\mc H_A$ is the Hilbert space of local quantum fields contained in a spacelike region $A$ of a Cauchy slice, $\Sigma$, and $B$ is the complement of that region, $B=\Sigma\setminus A$.

Because of the dense local entanglement structure of the quantum field theory vacuum, for virtually all (pure) states of the theory the entanglement entropy of a subregion $A$ is dominated by the short distance correlations shared between fields just inside and outside of $A$, straddling its boundary, $\pa A$. This leads to an `area law' entanglement entropy
\beq\label{eq:Sarealaw}
    S_A\approx\frac{\text{Area}[\pa A]}{\epsilon^{d-2}}+\ldots
\eeq
where $\epsilon$ is a short-distance cutoff on the modes close to $\pa A$ (this can be made more precise by discretizing the theory to a lattice; $\epsilon$ is then proportional to the lattice spacing).

Regardless of this divergence (which would na\"ively make the entanglement entropy an ill-defined, empty quantity for QFTs) the strict inequalities satisfied by quantum information measures (positivity, monotonicity under inclusions, etc.) make $S_A$ a useful tool for analyzing and constraining QFTs. Additionally the resemblance of \eqref{eq:Sarealaw} to the Bekenstein-Hawking entropy of black-holes
\beq
    S_\text{BH}=\frac{\text{Area}[\text{horizon}]}{4G_N}
\eeq
suggesting a potential interplay between quantum information and gravitational physics, which we will only briefly explore in Lecture \ref{sect:l6}.

A closely related object to the entanglement entropy is the relative entropy between two states, $\rho$ and $\sigma$, given by
\beq
    S_\text{rel}(\rho||\sigma)=\Tr\left(\rho\log\rho\right)-\Tr\left(\rho\log\sigma\right)~.
\eeq
The relative entropy is a measure of {\it distinguishability}. It is positive $S_\text{rel}(\rho||\sigma)\geq 0$ with equality iff the two states are equal.
\beq
    S_\text{rel}(\rho||\sigma)=0\qquad\Leftrightarrow\qquad \rho=\sigma~.
\eeq
Applied to reduced states it has a number of nice properties. For one it is {\it monotonic} under inclusions,
\beq\label{eq:Srelmono}
    S_\text{rel}(\rho_A||\sigma_A)\geq S_\text{rel}(\rho_{A'}||\sigma_{A'})~,\qquad \mc D(A')\subseteq \mc D(A)~,
\eeq
where $\mc D(\cdot)$ denotes the domain of dependence. This monotonicity simply expresses that if we trace out more of a region, $A\rightarrow A'$, we have less information by which to distinguish $\rho$ and $\sigma$.

\subsection*{Modular Hamiltonians}

An important object in the study of entanglement entropy and perhaps our first instance of a direct connection to energy density is the {\it modular Hamiltonian} of a reduced state. For a state $\rho$ reduced on a spacelike subregion $A$, the modular Hamiltonian is defined as
\beq
    K_A:=-\frac{1}{2\pi}\log\rho_A~.
\eeq
This definition is very much physically motivated by the fact that the most well known mixed state, the thermal or the Gibbs state, can be written as
\beq
    \rho_\text{Gibbs}=e^{-\beta H}~,\qquad H\text{ is the Hamiltonian.}
\eeq
So $K_A$ is the ``Hamiltonian'' of the general quantum mixed state $\rho_A$ at temperature $(2\pi)^{-1}$. However, unlike the Gibbs state, we have no right or reason to expect $K_A$ to be integral of a local stress tensor over a spacelike subregion (as is the Hamiltonian \eqref{eq:Ham}): it is a non-analytic functional of a reduced quantum state. Nevertheless, much like the Hamiltonian of a Gibbs state, the modular Hamiltonian is a formally positive operator,
\beq
    K_A\geq 0~,
\eeq
owing to the fact that density matrix is a normalized, $\Tr_{\mc H_A}\rho_A=1$, and positive operator. I have said ``formally" above because $K_A$ is divergent in most states due to a pile-up of short-distance modes near the boundary of $A$. This is essentially the same divergence in the entanglement entropy itself, \eqref{eq:Sarealaw}. A much better behaved object is the so-called {\it full modular Hamiltonian} which comes from subtracting the modular Hamiltonian of the complement region:
\beq
    \mathbf K_A:=K_A\otimes 1_B-1_A\otimes K_{B}~.
\eeq

Amazingly there are special situations in which modular Hamiltonian takes a local form (i.e. expressed as the integral of a local operator). The simplest example is the vacuum state reduced to a half-space of Minkowksi space, $A=\{\vec x\big|x^1>0\}$,
\beq\label{eq:RindModHam}
    K_A=\int_{A}\dd^{D}x\,x^1\,\hat T_{00}(t=0,\vec x)~,
\eeq
which is the generator of boosts in the Rindler wedge, a result known as the Bisognano-Wichmann theorem \cite{Bisognano:1976za}. This is sensible, the flow of $K_A$ is an accelerating observer who experiences a Rindler horizon and so has ``lost'' access to half of Minkowski space. Her state is mixed -- it is thermal with Unruh radiation.

The full modular Hamiltonian of the Rindler wedge is simply the double-sided boost generator:
\beq
    \mathbf{K}_\text{Rind.}=\int_\Sigma\,\dd^D\vec x\,x^1\,\hat T_{00}(t=0,\vec x)~.
\eeq
Note that because of the flip of sign of $x^1$, flows of $\mathbf{K}_\text{Rind.}$ take opposite directions in left and right Rindler wedges.

In CFTs we can generalize this result by acting global conformal transformations on the region $A$. Since conformal transformations map circles / lines to circles / lines, the half-space maps to a ball region, $B_R=\{\vec x~\big|~\abs{\vec x}\leq R\}$. The stress tensor is a primary operator under conformal transformations which act locally on it. The result is the modular Hamiltonian for the CFT vacuum reduced to a ball region, $A=B_R$ \cite{Casini:2011kv}:
\beq
    K_A=\int_{A}\dd^D\vec x\,\frac{R^2-r^2}{2R}\hat T_{00}(t=0,\vec x)~.
\eeq
The interesting aspect of these expressions is that the information content of the quantum vacuum (its entanglement spectrum) is an integrated energy density. Of course, the half-space and the ball region are very special subregions, however in what follows we will employ the fact that the edge of the causal domain of dependence for any region, on small enough scales, looks like the Rindler wedge.

\subsection*{Proof of the ANEC}

As a primary example of the utility of relating energy to quantum information, let us review a nifty proof of the ANEC developed by Faulkner, Leigh, Parrikar, and Wang \cite{Faulkner:2016mzt}. The key tools of this proof are the monotonicity of the relative entropy, \eqref{eq:Srelmono}, the expression of the vacuum modular Hamiltonian for a Rindler wedge, \eqref{eq:RindModHam}, and a relation between relative entropy and the vacuum modular Hamiltonian given by
\beq
    S_\text{rel}\left(\rho^{(\psi)}_A||\rho^{(\Omega)}_A\right)=\frac{1}{2\pi}\Delta\langle K_A^{(\Omega)}\rangle-\Delta S_A~,
\eeq
where
\beq
    \Delta\langle K_A^{(\Omega)}\rangle=\langle K_A^{(\Omega)}\rangle_\psi-\langle K^{(\Omega)}_A\rangle_\Omega~,\qquad \Delta S_A=S_A[\rho^{(\psi)}]-S_A[\rho^{(\Omega)}]~.
\eeq
which follows from the definitions of $K_A$ and $S_A$, plus simple algebra. Because the entanglement entropy of pure states is complementary, $S_A[|\psi\rangle\langle\psi|]=S_B[|\psi\rangle\langle\psi|]$ subtracting this relation for any two pure reference states implies the operator monotonicity relation of the full modular Hamiltonian
\beq\label{eq:modHamMono}
    \mathbf{K}_A-\mathbf{K}_{A'}\geq 0~,\qquad \mc D(A')\subseteq\mc D(A)~.
\eeq
The authors consider a deformation of the half-space $A=\{x^\mu~\big|~x^0=0,~x^1\geq 0\}$ to a region $A'=\{\tilde x^\mu~\big|~\tilde x^0=\zeta^0(x^1,\vec x_\perp)~,\tilde x^1=x^1+\zeta^1(x^1,\vec x_\perp)\geq \zeta^1(x^1,\vec x_\perp)\}$. In the limit of infintesimal $\zeta$, the authors show
\begin{align}
    K_{A'}-K_A\approx&-\int\dd^{d-2}\vec x_\perp\int_0^\infty\,\dd x^+\,\zeta^+\hat T_{++}(x^+,x^-=0,\vec x_\perp)\nonumber\\
    &+\int\dd^{d-2}\vec x_\perp\int_{-\infty}^0\,\dd x^-\,\zeta^-\hat T_{--}(x^+,x^-=0,\vec x_\perp)~.
\end{align}
A similar result can be established for complement region $B$ under this deformation and so the full modular deformed Hamiltonian takes the form
\begin{align}
    \mathbf{K}_{A'}-\mathbf{K}_A\approx&-\int\dd^{d-2}\vec x_\perp\int_{-\infty}^\infty\,\dd x^+\,\zeta^+\hat T_{++}(x^+,x^-=0,\vec x_\perp)\nonumber\\
    &+\int\dd^{d-2}\vec x_\perp\int_{-\infty}^\infty\,\dd x^-\,\zeta^-\hat T_{--}(x^+=0,x^-,\vec x_\perp)~.
\end{align}
Since the $\mc D(A')\subseteq\mc D(A)$ for $\zeta^+>0$ and $\zeta^-<0$ and otherwise arbitrary, the monotonicity of the full modular Hamiltonian \eqref{eq:modHamMono} implies the positivity of the ANEC as an operator,
\beq
    \int_{-\infty}^\infty\dd x^+T_{++}(x^+,x^-=0,\vec x_\perp)\geq 0~.
\eeq

\subsection*{QNEC}

In the last subject of this lecture we will relate a {\it local bound} on the null energy to quantum information. This is known now as the {\it quantum null energy condition} (QNEC) \cite{Bousso:2015mna}. In very rough terms, if a single variation of the entangling surface of a region yields an integrated bound on null energy (the ANEC), then the second variation results in a local bound. We will introduce it somewhat unmotivated at this point (the original motivation arose from a much more general bound, the quantum focussing conjecture \cite{Bousso:2015mna}, that for pedagogical reasons we will discuss next lecture). 

The QNEC is as follows. Consider a point $ x=(x^+,x^-,\vec {x}_\perp)$, a Cauchy slice containing $\bar x$ such that there exists a closed codimension-2, $\Gamma$, containing $x$ such that both of future null expansions vanish at $\bar x$ (we say that $\Gamma$ is {\it marginal} at $x$). $\Gamma$ encloses a spacelike region, $A\subset \Sigma$. For a state $|\psi\rangle$ in the Hilbert space on $\Sigma$ we can consider its entanglement entropy upon reduction to $A$, $S_A[\rho]$. For a null vector, $k^\mu$, pointing towards the interior of $A$ at $x$ we consider variations of $S_A[\rho]$ under an infinitesimal null deformation $\delta\lambda(x)$ of $\Gamma$ in the direction of $k$ localized to $\bar x$. See Figure \ref{fig:QNECsetup}.

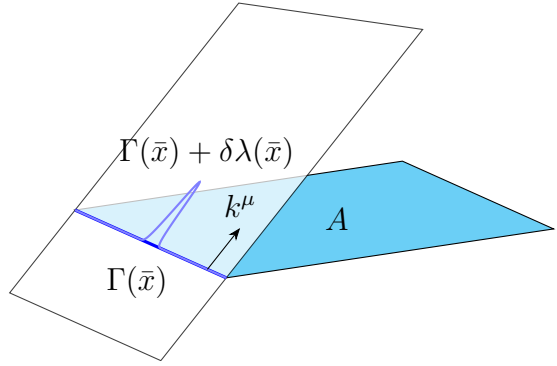
\begin{figure}[h!]
    \centering
    \begin{tikzpicture}[3d view,perspective]
    \filldraw[draw=black,fill=cyan!50] (1,0,1) -- (1,4,1) -- (6,4,1) -- (6,0,1) -- (1,0,1);
    \filldraw[draw=black,fill=white,opacity=.75] (0,0,0) -- (0,4,0) -- (3.5,4,3.5) -- (4,0,4) -- (0,0,0);
    \draw[very thick, blue] (1,0,1) -- (1,4,1);
    \draw[thick, blue!50, smooth]
    plot[domain=0:4,samples=145] ({.75*exp(-(\x-2)^2/.01)+1},{\x},{.75*exp(-(\x-2)^2/.01)+1});
    \node[below] at (.85,2.1,.85) {$\Gamma(\bar x)$};
    \node[above] at (1.8,1.9,1.8) {$\Gamma(\bar x)+\delta\lambda(\bar x)$};
    \node[right] at (3.5,2,1) {$A$};
    \node[above] at (1.5,.5,1.5) {$k^\mu$};
    \draw[-Stealth] (1,.5,1) -- (1.5,.5,1.5);
\end{tikzpicture}
\caption{The QNEC set-up. The codimension-2 surface $\Gamma$ is an entangling surface for $A$ with vanishing null expansion at $\bar x$. We consider the second variation of $S_A$ with respect to a small deformation of $\Gamma$ in a null direction pointing towards the interior of $A$.}\label{fig:QNECsetup}
\end{figure}

The null energy density in $|\psi\rangle$ is bounded below by the second null variation of $S_A$ at $x$:
\beq\label{eq:QNEC}
    \langle k^\mu k^\nu T_{\mu\nu}(\bar x)\rangle_\psi\geq \frac{1}{2\pi\sqrt{\hat h(\bar x)}}\frac{\delta^2}{\delta\lambda(x)^2}S_A[\rho^{(\psi)}]\Big|_{\delta\lambda=0}~.
\eeq
We have talked at length in past lectures about how local energy densities are unbounded in QFT and the QNEC is not a free lunch in this sense. The right-hand side is unbounded over all states in QFT as well. The key benefit of the QNEC is that it is the sharpest relation of a local energy density to a quantum information quantity and as such brings the rich and rigid structure that quantum information theory must satisfy to the toolkit of constraining energy densities. This is even sharper for interacting QFTs, where there is an accumulation of evidence that the QNEC is in fact, {\it saturated} with an equality appearing in \eqref{eq:QNEC} \cite{Leichenauer:2018obf,Balakrishnan:2019gxl}. (The QNEC is known to not be saturated in free theories\footnote{It is worth thinking to one's self what is different about free theories vs. interacting theories that makes this so!}). More importantly, the QNEC has been {\it proven} for both free \cite{Bousso:2015wca} and interacting field theories \cite{Balakrishnan:2017bjg} in Minkowski spacetime putting it (amongst the ANEC) in the very rare set of lower bounds that are accepted as simply just true.

The proof of the QNEC for interacting field theories, \cite{Balakrishnan:2017bjg}, is a {\it tour de force} and an entire lecture series could be spent establishing the wealth of tools that go into its proof (as originally stated). We will instead give a summarized road-map of the free field QNEC proof \cite{Bousso:2015wca} (or a simplified variant \cite{Balakrishnan:2019gxl}) based upon the properties of free fields quantized on a lightsheet we established in Lecture \ref{sect:l3}.

\subsubsection*{A free field proof of QNEC}

We start to restating the QNEC as the positivity of the relative entropy under second null deformations
\begin{align}
    \frac{\delta^2}{\delta\lambda(x)^2}S_\text{rel}\left(\rho^{(\psi)}_A||\rho^{(\Omega)}_A\right)=&\frac{\delta^2}{\delta\lambda(x)^2}\left(\frac{1}{2\pi}\Delta \langle K_A^{(\Omega)}\rangle-\frac{\delta^2}{\delta\lambda^2}\Delta S_A\right)\nonumber\\
    =&\frac{1}{2\pi}\sqrt{\hat h(x)}\langle :T_{++}(x):\rangle_\psi-\frac{\delta^2}{\delta\lambda(x)^2}S_A[\rho^{(\psi)}]\geq 0~,
\end{align}
where in the moving from the first to the second line we have used the variation of the half-space modular Hamiltonian as in the derivation of the ANEC and we have also used that because of translation invariance of the vacuum, $\frac{\delta}{\delta\lambda}S_A[\rho^{(\Omega)}]=0$. The inequality in the second line is not a quantum information inequality: it is a restatement of the QNEC and we must still prove it. While the first variation of $S_\text{rel}$ is sign-definite under monotonicity, establishing the inequality for the second variation is the non-trivial part of the QNEC proof. 

We now consider a $A$ be a half-space of a Cauchy slice of Minkowksi spacetime and $\mc L$ be the lightsheet of null rays shot to the future and past of $\pa A$.\footnote{This lightsheet is infinitely extendable because $\pa A$ is marginal and Minkowski spacetime is flat so the Raychaudhuri equation \eqref{eq:nullRayEq} guarantees that the expansion remains zero in both directions.} We now quantize free fields on this light-sheet in the same fashion as Lecture \ref{sect:l3}. In particular we break $\mc L$ into pencils of transverse width $a$ (the infinitesimal volume element $\sqrt{\hat h}$ is then replaced by $a^{d-2}$). See Figure \ref{fig:FFQNEC} for a cartoon.

\begin{figure}[h!]
    \centering
    \begin{tikzpicture}[3d view,perspective]
    \filldraw[draw=black,fill=cyan!50] (1,0,1) -- (1,4,1) -- (6,4,1) -- (6,0,1) -- (1,0,1);
    \filldraw[draw=black,fill=white,opacity=.75] (0,0,0) -- (0,4,0) -- (3.5,4,3.5) -- (4,0,4) -- (0,0,0);
    \fill[magenta,opacity=.5] (0,1.75,0) -- (0,2.25,0) -- (3.72,2.25,3.72) -- (3.79,1.75,3.79) -- (0,1.75,0);
    \draw[blue] (1,0,1) -- (1,4,1);
    \draw[thick,blue,rounded corners] (1,0,1)--(1,1.75,1) -- (1.5,1.75,1.5) -- (1.5,2.25,1.5) -- (1,2.25,1) -- (1,4,1);
    \node[above] at (1.8,1.9,1.8) {$\delta\lambda(\bar x)$};
    \node[below] at (.8,2.3,.8) {$\mathcal{P}$};
    \node[right] at (3.5,2,1) {$A$};
    \node[above] at (1.5,.5,1.5) {$k^\mu$};
    \node[above] at (3.8,2,3.8) {$a$};
    \draw[-Stealth] (1,.5,1) -- (1.5,.5,1.5);
\end{tikzpicture}
\caption{The set-up for the free field proof the QNEC. The null deformation is localized to a fixed pencil, $\mathcal P$, and for small enough transverse area can be regarded as a translation along that pencil generated by vacuum modular Hamiltonian.}\label{fig:FFQNEC}
\end{figure}
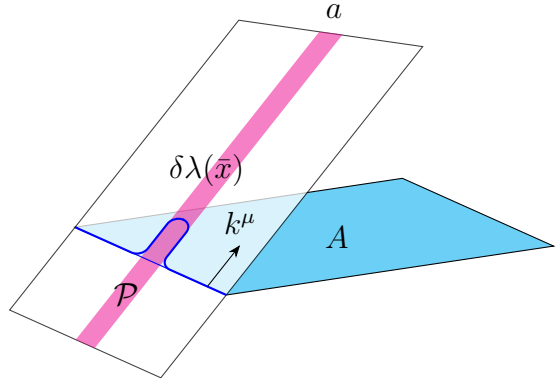

A key feature is that the state restricted to a specific pencil, $\mc P$, (namely the one containing null deformation) and reduced up to null cut of $\mc L$, is ``close'' to the vacuum for small $a$ \cite{Bousso:2015wca},
\beq
    \rho(\lambda)=\rho_{\mc P}^{(\Omega)}\otimes \rho_\text{rest}^{(\Omega)}+\sigma(\lambda)~,
\eeq
where $\sigma$ contains information about the excited state as well as the first non-trivial entanglement of $\mc P$ amongst the other pencils (here labelled ``rest"). $\sigma$ admits an expansion in $a$ with first nontrivial term coming from an insertion of a single particle operator on $\mc P$ entangled with the rest of the pencils,
\beq
    \sigma=a^{\frac{d-2}{2}}\sum_{ij}\int \dd r\dd\theta\,f_{ij}(r,\theta)\pa\phi(re^{i\theta})\otimes E_{ij}(\theta)+\ldots~,
\eeq
for some basis of density matrices, $E_{ij}$, of the "rest" system which w.l.o.g we can take to be eigenstates of the modular Hamiltonian of the ``rest" system with eigenvalue $\kappa_i$, $E_{ij}(\theta)=e^{\theta(\kappa_i-\kappa_j)}|i\rangle_\text{rest}\langle j|_\text{rest}$. There is a general formalism for the relative entropy of states ``close to the vacuum" in terms of the flow generated by the vacuum modular Hamiltonian itself (i.e. {\it modular flow}) \cite{Faulkner:2017tkh}
\beq
    S_\text{rel}\left(\rho^{(\psi)}(\lambda)||\rho^{(\Omega)}(\lambda)\right)=-\frac{1}{2}\int_{-\infty}^\infty\frac{\dd s}{4\sinh^2\left(\frac{s+i\epsilon}{2}\right)}\langle {\rho^{(\Omega)}}^{-1}\sigma e^{is\mathbf K^{(\Omega)}(\lambda)}{\rho^{(\Omega)}}^{-1}\sigma\rangle_{\Omega}+\ldots
\eeq
where $\mathbf{K}^{(\Omega)}$ the full modular Hamiltonian of $\rho_{\mc P}^{(\Omega)}\otimes\rho_\text{rest}^{(\Omega)}$ reduced up to $\lambda$. As we discussed earlier, the vacuum modular Hamiltonian is a generator of boosts which are translations along the null ray defining the pencil. Thus taking the null derivatives of $S_\text{rel}$ we find
\begin{align}
    \frac{\dd^2}{\dd\lambda^2}S_\text{rel}=\frac{a^{d-2}}{2}\int\dd s~e^{s}\int \dd^2 r_{1,2}\dd^2\theta_{1,2}&\langle (\pa^3\phi)(r_1e^{i\theta_1})(\pa\phi)(r_2e^{i\theta_2+s})\rangle_{\Omega_{\mc P}}\nonumber\\
    &\sum_{ij}\langle E_{ij}(\theta_1)E_{ji}(\theta_2-is)\rangle_{\text{rest}}~.
\end{align}
The $\langle \cdot\rangle_{\Omega_{\mc P}}$ piece can be evaluated as a chiral 1+1 CFT correlator
\beq
    \langle \pa\phi(z_1)\pa\phi(z_2)\rangle=\frac{1}{(z_1-z_2)}~,
\eeq
and the $\langle \cdot\rangle_\text{rest}$ is evaluated in terms of the vacuum modular energies, $\kappa_i$, of the ``rest" system. Ultimately we can massage this entirely into
\beq\label{eq:d2SrelFinal}
    \frac{\dd^2}{\dd\lambda^2}S_\text{rel}=\frac{a^{d-2}}{2}\sum_{ij}\abs{F^{(2)}_{ij}}^2\,e^{-\pi(\kappa_i+\kappa_j)}\mathsf G\left(\kappa_i-\kappa_j\right)~,
\eeq
with
\beq
    F_{ij}^{(m)}=\int\frac{\dd r}{r^m}r^{i(\kappa_i-\kappa_j)}f_{ij}^{(m)}(r)~,\qquad \mathsf G(\nu)=\frac{\pi \nu(1+\nu^2)}{\sinh\pi\nu}~.
\eeq
The expression \eqref{eq:d2SrelFinal} is inherently positive which then establishes the QNEC for free (and super-renormalizable) theories.

\section{Lecture 6: Gravity and quantum information}\label{sect:l6}

In this final lecture we combine the tools that we have built up in the previous lectures to review the recent developments in using quantum information to constrain semi-classical gravity and for insights on the low-energy nature of quantum gravity, more broadly.

The first hints of a connection between semi-classical gravity and information theory was the realization by Bekenstein and Hawking that black holes are thermodynamic objects with a temperature and an entropy \cite{Bekenstein:1974ax,Hawking:1974rv,Hawking:1975vcx}. Moreover this entropy is finite and given by the area of its event horizon in Planck units:
\beq\label{eq:SBH2}
    S_\text{BH}=\frac{\mc A_\text{horizon}}{4G_N}~.
\eeq
The unavoidable growth of the black hole horizon mirrors the unavoidable growth of the entropy, a statement known as the {\it second law of black hole thermodynamics.}

The interpretation of \eqref{eq:SBH2} has been a subject of continual discussion in theoretical high-energy physics. Given the `no-hair theorem' \cite{Israel:1967wq,Israel:1967za,Carter:1971zc,Hawking:1971vc} that black holes are determined entirely by a handful of macroscopic quantities  (mass, angular momentum, and electric and magnetic charge), a standard thermodynamic interpretation of $S_\text{BH}$ is that it is a count of the number of black hole microstates sharing the same no-hair variables. However this count of microstates becomes confusing in light of quantum gravity: we can imagine forming a black hole from the collapse of a single pure quantum state for instance. A modern interpretation, and one of growing consensus, is that $S_\text{BH}$ is a form of entanglement entropy of microscopic (quantum gravitational) degrees of freedom that we have lost access to due to horizon formation. This interpretation was greatly bolstered by Susskind and Uglum \cite{Susskind:1994sm} who argued for the interpretation of $S_\text{BH}$ in string theory as open strings ``cut'' by the horizon and by the subsequent calculation of Strominger and Vafa \cite{Strominger:1996sh} who reproduced the microstate counting of certain extremal black holes from counting brane configurations in string theory.

The connection between areas and quantum information takes a much deeper role than just their application to horizons as evidenced through holographic duality and in particular, the advent of the AdS/CFT correspondence. In the AdS/CFT correspondence, the entanglement entropy of a CFT state reduced upon a region $A$, is given by minimal areas extending into the bulk of AdS space and anchored at $A$, a result known as the {\it Ryu-Takayanagi formula} \cite{Ryu:2006bv}:
\beq\label{eq:RTform}
    S_A=\min_{\substack{\gamma\\\pa\gamma=A}}\frac{\mc A_\gamma}{4G_N}~.
\eeq

\subsection*{The generalized entropy}

The interpretation of the Bekenstein-Hawking entropy as an entanglement entropy, as well as indication by the Ryu-Takanagi formula of a relation between minimal areas and entanglement entropies leads us to an intriguing mechanism for regularizing the area law divergences \eqref{eq:Sarealaw} of entanglement entropy of quantum field theory when embedded into dynamical gravity. This mechanism was in fact first suggested by the arguments of Susskind and Uglum \cite{Susskind:1994sm}, who considered the quantity
\beq
    S_A+\frac{\mc A_{\pa A}}{4G_N}~,
\eeq
and showed that the area-law divergence \eqref{eq:Sarealaw} could be interpreted\footnote{This is not a vacuous or conjectural interpretation. A computation was done in linearized gravity to show that the one-loop contribution to coupling matched this divergence.} as a one-loop renormalization of Newton's constant, $G_N\rightarrow G_N^\text{ren}$. This suggests that the quantity known as the {\it generalized entropy} of a spacelike region $A$, defined as its boundary area in Planck units plus the entanglement of quantum fields reduced to $A$,
\beq\label{eq:Sgen}
    S_\text{gen}[A]=\frac{\mc A_{\pa A}}{4G_N}+S_A~,
\eeq
is a UV-finite quantity in quantum gravity. Note that this object really is defined for semi-classical gravity, as an effective field theory, with a dynamical graviton turned on: $S_A$ quantifies the entropy low-energy fluctuations (quantum fields, and gravitons) about a classical solution (spacetime geometry) in an unstated UV theory; amazingly, the low-energy avatar of the all the unknown UV degrees of freedom organize into a geometric area law with a finite (renormalized) coefficient.

This leads to general paradigm of promoting geometric quantities in semi-classical gravity to quantum informational quantities in an agnostic UV complete-theory of gravity by replacing codimension-2 areas with generalized entropies. Below we will explore some examples and consequences of this paradigm.

\subsection*{The generalized second law}

The second law of black hole thermodynamics, while classically true, is easily violated by quantum corrections due the realization (by Hawking) that black holes radiate quantum mechanically \cite{Hawking:1975vcx}. This realization is very analogous to the realization by Unruh that the vacuum of Minkowski space is thermal with respect to the vacuum of the Rindler wedge. In short, vacuum states generically look thermal outside of horizons. This radiation can escape, with low probability, to asymptotic infinity and when it does it carries energy away from the black hole, causing its mass and horizon area to shrink. Thus without any outside matter to source the black hole, $\mc A_\text{horizon}$ can in fact {\it decrease} over long time scales, violating the second law.

This was already realized in Hawking's paper \cite{Hawking:1975vcx} who noted that the {\it generalized second law} proposed one year earlier by \cite{Bekenstein:1974ax} is not violated by this scenario. The generalized second law is as follows:\\
\\
{\bf Generalized second law: }{\it The sum of the all black hole horizon areas (measured in Planck units) and the entropy of quantum fields exterior to those horizons never decreases.}\\
\\
In other words, the generalized entropy of the exterior of all black holes never decreases:
\beq\label{eq:gen2ndlaw}
    \dd S_\text{gen}[\text{BH}]\geq 0~.
\eeq
The generalized second law avoids the problem with Hawking radiation because the radiation comes in pairs. A small quanta of radiation that escapes to null infinity is created just outside of the horizon and comes entangled with a quanta just inside the horizon that falls into the singularity. While this creation decreases the horizon area, as mentioned above, the tracing out of the partner inside the horizon leads to an increase of the entanglement entropy of the quantum fields to the exterior of the horizon. The generalized second law is compatible with the radiation outside of causal horizons more broadly, e.g. the de Sitter and Rindler horizons \cite{Jacobson:2003wv}. Lastly, it has been shown that the generalized second law implies the achronal ANEC \cite{Wall:2009wi}.

The status of the generalized second law as a scientific law is still uncertain however. There are a variety of ``proofs'' starting from various assumptions (e.g. the covariant entropy bound). Most concretely, a proof of the generalized second law for quantum fields fields falling across general causal horizons was given by Wall \cite{Wall:2011hj} based upon certain assumptions of ultra-locality of the algebra of fields quantized on that horizon. These assumptions are directly verified for free and super-renormalizable QFTs as we have already encountered in Lectures \ref{sect:l3} and \ref{sect:l4}, however remains an assumption more generally.

\subsection*{The quantum focussing conjecture}

Motivated by the generalized second law, Bousso, Fisher, Leichenauer, and Wall proposed a generalization of the classical focussing theorem holding in semi-classical gravity, coined the {\it quantum focussing conjecture} \cite{Bousso:2015mna}. The coarse logic behind the QFC is that in quantum field theory and with dynamical gravity, neither $\frac{\mc A_{\pa A}}{4G_N}$ nor $S_A$ for a region are well defined due to quantum fluctuations. However the combination $S_\text{gen}[A]$ is. Similarly while the NEC ensures the classical focussing theorem
\beq
    \frac{\dd}{\dd\lambda}\theta\leq 0~,\qquad (R_{kk}\geq 0)~,
\eeq
this is violated in QFT by the need to subtract divergent quantum fluctuations (i.e. normal ordering). Their proposal is to replace the area in the expansion with the generalized entropy. To be specific, they define a {\it quantum expansion} of a closed codimension-2 surface, $\Gamma$, containing a point $x$ as
\beq
    \Theta[\Gamma;x]=\frac{4G_N}{\sqrt{\hat h(x)}}\frac{\delta}{\delta \lambda(x)}S_\text{gen}[A]\Big|_{\lambda=0}~,
\eeq
where much like the definition of the QNEC, $A$ is a spacelike region enclosed by $\Gamma$ and $\lambda(x)$ is a null deformation of $\Gamma$ pointing towards the interior of $A$ and localized at $x$. However unlike the QNEC we do not require $\Gamma$ to be marginal at $x$ in this definition. The quantum focussing conjecture is that $\Theta$ decreases along null deformations of $\Gamma$, i.e.
\beq\label{eq:QFC}
    \frac{\delta}{\delta\lambda(y)}\Theta[\Gamma;x]\leq 0~.
\eeq
Unlike the generalized second law, the QFC is conjectured to hold for any closed spacelike codimension-2 surface. Note that because $\Theta$ is non-local (depending on an entire codimension-2 surface, $\Gamma$) we do not have to apply the QFC at the same point that the quantum expansion as the argument of $\Theta$, i.e. $y\neq x$ is allowed. Thus the QFC yields an infinite amount of off-diagonal constraints that have no analog for the classical focussing theorem. This off-diagonal QFC is easily proven, however, (in fact in the same paper the QFC was introduced) because it only involves the entanglement entropy, $S_A$, contribution of $S_\text{gen}$. A quantum information bound known as {\it strong subadditivity} then implies 
\beq
    \frac{\delta}{\delta\lambda(y)}\Theta[\Gamma;x]\leq 0~,\qquad x\neq y~.
\eeq
The diagonal part of the QFC, involving the second variation of the generalized entropy at a point $x$, is much more non-trivial and remains unproven. There are no known counter-examples to the diagonal QFC. Moreover, the QFC seems to be the strongest known inequality applicable to semi-classical gravity.\footnote{In fact, it might even be too strong as originally stated: Shahbazi-Moghaddam has argued that the slightly weaker condition that the QFC inequality holds at locations where the quantum expansion itself vanishes is sufficient to derive all the known implications of the QFC \cite{Shahbazi-Moghaddam:2022hbw}.} It implies \cite{Bousso:2015mna} the {\it covariant entropy bound} \cite{Bousso:1999xy} (or colloquially, the ``Bousso bound"), the generalized second law along causal horizons \cite{Bousso:2015mna}, and in turn, the achronal ANEC along those horizons. In the context of AdS/CFT the QFC implies \cite{Akers:2016ugt} a property known as {\it entanglement nedge westing} (which is important for ensuring that holographic entanglement entropy is well behaved under subregion inclusions) \cite{Czech:2012bh,Wall:2012uf}, that the causal domain of a boundary subregion is contained in the entanglement wedge (an important component of bulk reconstruction) \cite{Engelhardt:2014gca}, and the principle of `no bulk shortcut' \cite{Akers:2016ugt}. Lastly, in any background the QFC implies the QNEC. Let us show this latter implication briefly.

We consider the diagonal contribution to the QFC and denote the null variation at $x$ by a prime, $'$. The quantum expansion is thus notated
\beq
    \Theta=\theta+\frac{4G_N}{\sqrt{\hat h}}S_A'~.
\eeq
The second null variation is then
\begin{align}
    \Theta'=&\theta'+\frac{4G_N}{\sqrt{\hat h}}\left(S_A''-S_A'\,\theta\right)\nonumber\\
    =&-\frac{1}{(d-2)}\theta^2-\sigma^2+\omega^2-8\pi G_N\langle T_{kk}\rangle_\psi+4G_N\left(S_A''-S_A'\,\theta\right)~,
\end{align}
where in the second line we replace $\theta'$ with the null Raychaudhuri equation and the $R_{kk}$ with the null stress-energy tensor assuming a consistent solution to the semi-classical Einstein equation. We now see that if $\Gamma$ is margigal at $x$ such that $\theta$ (and we additionally arrange for it to be shearless and irrotational, $\sigma_{\mu\nu}=\omega_{\mu\nu}=0$) then the QFC, $\Theta'\leq 0$ implies
\beq\label{eq:QFCtoQNEC}
\langle T_{kk}\rangle_\psi\geq \frac{1}{2\pi}S_A''~,
\eeq
which is a compact notation for the QNEC. What's fascinating about the above manipulation is that while QFC is very much a dynamical gravity inspired bound, all $G_N$'s have dropped from \eqref{eq:QFCtoQNEC}. At the time of its formulation this strongly hinted that the QNEC could be true and proven directly in QFT without dynamical gravity. The subsequent realization of such proofs gives strong credence to the validity of the QFC.

\subsection*{The quantum Penrose singularity theorem}

In this last section of the lectures, let us now use our quantum information toolkit to revisit the classic singularity theorems we saw before. We will go back to the beginning and revisit our first singularity theorem of Lecture \ref{sect:l2}, the Penrose singularity theorem. We will drop the assumption of the NEC, and instead utilize the generalized second law as the key physical input of the theorem. This theorem is due to Aron Wall \cite{Wall:2010jtc}.

A crux of Wall's theorem is that while the generalized second law is supposed (as an assumption of the theorem) to hold on all causal horizons, it certainly does not apply to generic null surfaces and there are known counter-examples \cite{Wall:2011hj}. In turn, if a null surface $\mc L$ has a decreasing generalized entropy it cannot be a causal horizon. This is subsequently used to show that certain null surfaces with decreasing $S_\text{gen}$ must terminate.

A key concept appearing Wall's quantum singularity theorem is the notion of a {\it quantum trapped surface}. Analogous to a trapped surface, a quantum trapped surface is one of negative quantum expansion. Namely, let $\Gamma$ be a codimension-2 closed spacelike submanifold of a Cauchy slice $\Sigma$. $\Gamma$ encloses a subregion $A\subset\Sigma$. $\Gamma$ is quantum trapped if for each point on $\Gamma$ and null generator future pointing and towards the exterior of $A$, the quantum expansion is negative, i.e. $S_\text{gen}$ is decreasing.

\begin{theorem}[Quantum singularity theorem]
    Let $M$ be a globally hyperbolic spacetime with non-compact Cauchy slices and a quantum trapped surface $\Gamma$ such that in a neighborhood of $\Gamma$ the semi-classical approximation holds (that is $\Theta$ is well approximated by the classical expansion, $\theta$). Then the generalized second law implies that $M$ is null geodesically incomplete.
\end{theorem}

The proof of this theorem essentially follows from the above comment that a null surface not violating the generalized second law cannot be a causal horizon and so either a null geodesic along which $S_\text{gen}$ is decreasing must terminate or, if it is extendible, it must exit $\mc L$. From here the incompleteness essentially follows Penrose's reasoning: this second scenario cannot happen due to global hyperbolicity and non-compactness of the Cauchy slice.

\subsubsection*{Applications}

In \cite{Wall:2010jtc}, Wall proposes various scenarios that the quantum singularity theorem applies to, beyond the standard prediction of singularity formation through black hole collapse. For example, he demonstrates that the generalized second law implies that a spacetime with non-compact Cauchy slices that are expanding (and semi-classical) originate from a Big Bang singularity a finite proper time to the past. Additionally the generalized second law restricts the formation of `baby universes' inside of black hole horizons as well periods of `restarted inflation' in asymptotically flat Minkwoski spacetimes. Further applications of the generalized second law in \cite{Wall:2010jtc} rule out `warp drives,' time machines based on closed time-like curves, and traversable wormholes.

\subsubsection*{Away from the semi-classical limit}

Several steps in establishing the quantum singularity theorem rely on a `semi-classical' approximation of the initial condition: the generalized entropy and the quantum expansion determined by the classical area and the classical expansion, respectively. It has been pointed out that there are scenarios where quantum fields are well-behaved yet allow for non-negligible fluctuations of the quantum expansion that act as a loopholes to Wall's theorem. Thus the quantum singularity theorem as formulated by Wall still allows for some special singularity avoidance scenarios, e.g. a bouncing cosmology \cite{Veneziano:2000pz,Khoury:2001wf,Graham:2017hfr}. This loophole has recently been closed by Bousso \cite{Bousso:2025xyc} by retooling Wall's theorem using a new quantum information quantity as known as the `conditional max entropy' as its starting point. Recently a singularity theorem was proposed by Engelhardt and Nagar relaxing the assumption of global hyperbolicity, allowing for Penrose diagrams in which spatial slices can undergo topology change, such as that of a fully evaporated black hole \cite{Engelhardt:2026qqs}.

\section{Conclusion}\label{sect:disc}

In these lecture notes we summarized a surface level review of classical and quantum energy conditions and their role in constraining solutions in semi-classical gravity. The aim of these lecture notes is not necessarily to make you an expert in any one given topic here, but instead to give you a foundation of understanding for the topic and to prepare you for engaging with ongoing research in these directions.

In that vein there are several interesting open directions to pursue. As emphasized in subsection \ref{sec:QEIs}, QEIs such as the ANEC have uses in their own right as QFT statements: they place bounds on free parameters of a theory and constrain renormalization group flows. Moreover, the both the ANE and the modular Hamiltonian fall into a broader class of light-ray operators composed of the null stress tensor integrated with a polynomial of a null coordinate and which display an interesting algebraic structure organizing degrees of freedom at asymptotic infinity \cite{Kravchuk:2018htv,Kologlu:2019mfz,Belin:2020lsr,Cordova:2018ygx}. More generally exploring which QEIs (even state-dependent ones) can be proven directly in QFTs either through CFT techniques, or through quantum information techniques \cite{Fliss:2025zhu} and to what degree they constrain the landscape of QFTs is a rich and worthwhile endeavor.

Separately, it is interesting to push on energy inequalities that directly incorporate semi-classical gravitational effects, such as the SNEC \eqref{eq:gravSNEC} and the QFC \eqref{eq:QFC}. While it is unexpected that there will be definitive proofs of such inequalities (precisely because they directly incorporate dynamical gravity in their formulation), testing their validity in known models of quantum gravity, e.g. holographic setups, as well as delineating their implications and their mutual dependencies are usual inputs. Further developing new ``quantum geometric theorems'' incorporating such bounds as ingredients remains an active area of research.

\section*{Acknowledgments}
I thank the organizers of the XXI Modave Summer School in Mathematical Physics for the invitation to lecture on this subject. I also thank Ben Freivogel, Eleni Kontou, Diego Pardo Santos, and Andrew Rolph for discussions and collaborations on topics appearing in these lectures. 
This work was supported partially by STFC consolidated grants ST/T000694/1 and ST/X000664/1, partially by Simons Foundation Award number 620869, and partially by FNRS MISU grant 40024018 ``Pushing Horizons
in Black Hole Physics.''

\appendix
\numberwithin{equation}{section}

\bibliographystyle{ytphys}
\bibliography{sources.bib}

@misc{WittenTalk,
  author       = {Witten, Edward},
  title        = {Some Comments on Energy Inequalities},
  howpublished = {Talk at IAS Online Workshop on Qubits and Black Holes},
  year         = {2020},
  note         = {Recording available at www.youtube.com/watch?v=0Oh-Kmy-mx0.},
  url          = {www.youtube.com/watch?v=0Oh-Kmy-mx0}
}

@article{Kologlu:2019mfz,
    author = "Kologlu, Murat and Kravchuk, Petr and Simmons-Duffin, David and Zhiboedov, Alexander",
    title = "{The light-ray OPE and conformal colliders}",
    eprint = "1905.01311",
    archivePrefix = "arXiv",
    primaryClass = "hep-th",
    reportNumber = "CALT-TH 2019-013, CERN-TH-2019-055",
    doi = "10.1007/JHEP01(2021)128",
    journal = "JHEP",
    volume = "01",
    pages = "128",
    year = "2021"
}

@article{Belin:2020lsr,
    author = "Belin, Alexandre and Hofman, Diego M. and Mathys, Gr{\'e}goire and Walters, Matthew T.",
    title = "{On the stress tensor light-ray operator algebra}",
    eprint = "2011.13862",
    archivePrefix = "arXiv",
    primaryClass = "hep-th",
    reportNumber = "CERN-TH-2020-200",
    doi = "10.1007/JHEP05(2021)033",
    journal = "JHEP",
    volume = "05",
    pages = "033",
    year = "2021"
}

@article{Kravchuk:2018htv,
    author = "Kravchuk, Petr and Simmons-Duffin, David",
    title = "{Light-ray operators in conformal field theory}",
    eprint = "1805.00098",
    archivePrefix = "arXiv",
    primaryClass = "hep-th",
    reportNumber = "CALT-TH 2018-018",
    doi = "10.1007/JHEP11(2018)102",
    journal = "JHEP",
    volume = "11",
    pages = "102",
    year = "2018"
}

@article{Cordova:2018ygx,
    author = "C{\'o}rdova, Clay and Shao, Shu-Heng",
    title = "{Light-ray Operators and the BMS Algebra}",
    eprint = "1810.05706",
    archivePrefix = "arXiv",
    primaryClass = "hep-th",
    doi = "10.1103/PhysRevD.98.125015",
    journal = "Phys. Rev. D",
    volume = "98",
    number = "12",
    pages = "125015",
    year = "2018"
}

@article{Fliss:2025zhu,
    author = "Fliss, Jackson R. and Rolph, Andrew",
    title = "{Curious QNEIs from QNEC: New Bounds on Null Energy in Quantum Field Theory}",
    eprint = "2510.26247",
    archivePrefix = "arXiv",
    primaryClass = "hep-th",
    month = "10",
    year = "2025"
}

@article{Shahbazi-Moghaddam:2022hbw,
    author = "Shahbazi-Moghaddam, Arvin",
    title = "{Restricted quantum focusing}",
    eprint = "2212.03881",
    archivePrefix = "arXiv",
    primaryClass = "hep-th",
    doi = "10.1103/PhysRevD.109.066023",
    journal = "Phys. Rev. D",
    volume = "109",
    number = "6",
    pages = "066023",
    year = "2024"
}

@article{Hawking:1971vc,
    author = "Hawking, S. W.",
    title = "{Black holes in general relativity}",
    doi = "10.1007/BF01877517",
    journal = "Commun. Math. Phys.",
    volume = "25",
    pages = "152--166",
    year = "1972"
}

@article{Carter:1971zc,
    author = "Carter, B.",
    title = "{Axisymmetric Black Hole Has Only Two Degrees of Freedom}",
    doi = "10.1103/PhysRevLett.26.331",
    journal = "Phys. Rev. Lett.",
    volume = "26",
    pages = "331--333",
    year = "1971"
}

@article{Israel:1967za,
    author = "Israel, Werner",
    title = "{Event horizons in static electrovac space-times}",
    doi = "10.1007/BF01645859",
    journal = "Commun. Math. Phys.",
    volume = "8",
    pages = "245--260",
    year = "1968"
}

@article{Israel:1967wq,
    author = "Israel, Werner",
    title = "{Event horizons in static vacuum space-times}",
    doi = "10.1103/PhysRev.164.1776",
    journal = "Phys. Rev.",
    volume = "164",
    pages = "1776--1779",
    year = "1967"
}

@article{Hawking:1974rv,
    author = "Hawking, S. W.",
    title = "{Black hole explosions}",
    doi = "10.1038/248030a0",
    journal = "Nature",
    volume = "248",
    pages = "30--31",
    year = "1974"
}

@article{Page:1981aj,
    author = "Page, Don N. and Geilker, C. D.",
    title = "{Indirect Evidence for Quantum Gravity}",
    reportNumber = "PRINT-81-0221 (PENN-STATE)",
    doi = "10.1103/PhysRevLett.47.979",
    journal = "Phys. Rev. Lett.",
    volume = "47",
    pages = "979--982",
    year = "1981"
}

@article{Hawking:1970zqf,
    author = "Hawking, S. W. and Penrose, R.",
    title = "{The Singularities of gravitational collapse and cosmology}",
    doi = "10.1098/rspa.1970.0021",
    journal = "Proc. Roy. Soc. Lond. A",
    volume = "314",
    pages = "529--548",
    year = "1970"
}

@book{Feynman:1996kb,
    author = "Feynman, R. P.",
    editor = "Morinigo, F. B. and Wagner, W. G. and Hatfield, B.",
    title = "{Feynman lectures on gravitation}",
    doi = "10.1201/9780429502859",
    isbn = "978-0-429-50285-9",
    year = "1996"
}

@article{Reeh:1961ujh,
    author = "Reeh, H. and Schlieder, S.",
    title = {{Bemerkungen zur unit{\"a}r{\"a}quivalenz von lorentzinvarianten feldern}},
    doi = "10.1007/BF02787889",
    journal = "Nuovo Cim.",
    volume = "22",
    number = "5",
    pages = "1051--1068",
    year = "1961"
}

@article{Iizuka:2025xnd,
    author = "Iizuka, Norihiro and Ishibashi, Akihiro and Maeda, Kengo and Nakayama, Haruki and Nishioka, Tatsuma",
    title = "{Energy Conditions and Quantum Information}",
    eprint = "2509.01286",
    archivePrefix = "arXiv",
    primaryClass = "hep-th",
    reportNumber = "NU-QG-10, OU-HET-1283",
    month = "9",
    year = "2025"
}

@article{EINSTEIN1936349,
	author = {Albert Einstein},
	doi = {https://doi.org/10.1016/S0016-0032(36)91047-5},
	issn = {0016-0032},
	journal = {Journal of the Franklin Institute},
	number = {3},
	pages = {349-382},
	title = {Physics and reality},
	url = {https://www.sciencedirect.com/science/article/pii/S0016003236910475},
	volume = {221},
	year = {1936}}

@article{Graham:2017hfr,
    author = "Graham, Peter W. and Kaplan, David E. and Rajendran, Surjeet",
    title = "{Born again universe}",
    eprint = "1709.01999",
    archivePrefix = "arXiv",
    primaryClass = "hep-th",
    doi = "10.1103/PhysRevD.97.044003",
    journal = "Phys. Rev. D",
    volume = "97",
    number = "4",
    pages = "044003",
    year = "2018"
}

@article{Khoury:2001wf,
    author = "Khoury, Justin and Ovrut, Burt A. and Steinhardt, Paul J. and Turok, Neil",
    title = "{The Ekpyrotic universe: Colliding branes and the origin of the hot big bang}",
    eprint = "hep-th/0103239",
    archivePrefix = "arXiv",
    doi = "10.1103/PhysRevD.64.123522",
    journal = "Phys. Rev. D",
    volume = "64",
    pages = "123522",
    year = "2001"
}

@inproceedings{Veneziano:2000pz,
    author = "Veneziano, G.",
    title = "{String cosmology: The Pre - big bang scenario}",
    booktitle = "{71st Les Houches Summer School: The Primordial Universe}",
    eprint = "hep-th/0002094",
    archivePrefix = "arXiv",
    reportNumber = "CERN-TH-2000-042",
    doi = "10.1007/3-540-45334-2_12",
    pages = "581--628",
    year = "2000"
}

@article{Engelhardt:2014gca,
    author = "Engelhardt, Netta and Wall, Aron C.",
    title = "{Quantum Extremal Surfaces: Holographic Entanglement Entropy beyond the Classical Regime}",
    eprint = "1408.3203",
    archivePrefix = "arXiv",
    primaryClass = "hep-th",
    doi = "10.1007/JHEP01(2015)073",
    journal = "JHEP",
    volume = "01",
    pages = "073",
    year = "2015"
}

@article{Wall:2012uf,
    author = "Wall, Aron C.",
    title = "{Maximin Surfaces, and the Strong Subadditivity of the Covariant Holographic Entanglement Entropy}",
    eprint = "1211.3494",
    archivePrefix = "arXiv",
    primaryClass = "hep-th",
    doi = "10.1088/0264-9381/31/22/225007",
    journal = "Class. Quant. Grav.",
    volume = "31",
    number = "22",
    pages = "225007",
    year = "2014"
}

@article{Czech:2012bh,
    author = "Czech, Bartlomiej and Karczmarek, Joanna L. and Nogueira, Fernando and Van Raamsdonk, Mark",
    title = "{The Gravity Dual of a Density Matrix}",
    eprint = "1204.1330",
    archivePrefix = "arXiv",
    primaryClass = "hep-th",
    doi = "10.1088/0264-9381/29/15/155009",
    journal = "Class. Quant. Grav.",
    volume = "29",
    pages = "155009",
    year = "2012"
}

@article{Akers:2016ugt,
    author = "Akers, Chris and Koeller, Jason and Leichenauer, Stefan and Levine, Adam",
    title = "{Geometric Constraints from Subregion Duality Beyond the Classical Regime}",
    eprint = "1610.08968",
    archivePrefix = "arXiv",
    primaryClass = "hep-th",
    month = "10",
    year = "2016"
}

@article{Bousso:1999xy,
    author = "Bousso, Raphael",
    title = "{A Covariant entropy conjecture}",
    eprint = "hep-th/9905177",
    archivePrefix = "arXiv",
    reportNumber = "SU-ITP-99-23",
    doi = "10.1088/1126-6708/1999/07/004",
    journal = "JHEP",
    volume = "07",
    pages = "004",
    year = "1999"
}

@article{Jacobson:2003wv,
    author = "Jacobson, Ted and Parentani, Renaud",
    title = "{Horizon entropy}",
    eprint = "gr-qc/0302099",
    archivePrefix = "arXiv",
    doi = "10.1023/A:1023785123428",
    journal = "Found. Phys.",
    volume = "33",
    pages = "323--348",
    year = "2003"
}

@article{Bekenstein:1974ax,
    author = "Bekenstein, Jacob D.",
    title = "{Generalized second law of thermodynamics in black hole physics}",
    doi = "10.1103/PhysRevD.9.3292",
    journal = "Phys. Rev. D",
    volume = "9",
    pages = "3292--3300",
    year = "1974"
}

@article{Hawking:1975vcx,
    author = "Hawking, S. W.",
    editor = "Gibbons, G. W. and Hawking, S. W.",
    title = "{Particle Creation by Black Holes}",
    doi = "10.1007/BF02345020",
    journal = "Commun. Math. Phys.",
    volume = "43",
    pages = "199--220",
    year = "1975",
    note = "[Erratum: Commun.Math.Phys. 46, 206 (1976)]"
}

@article{Ryu:2006bv,
    author = "Ryu, Shinsei and Takayanagi, Tadashi",
    title = "{Holographic derivation of entanglement entropy from AdS/CFT}",
    eprint = "hep-th/0603001",
    archivePrefix = "arXiv",
    reportNumber = "NSF-KITP-06-11",
    doi = "10.1103/PhysRevLett.96.181602",
    journal = "Phys. Rev. Lett.",
    volume = "96",
    pages = "181602",
    year = "2006"
}

@article{Susskind:1994sm,
    author = "Susskind, Leonard and Uglum, John",
    title = "{Black hole entropy in canonical quantum gravity and superstring theory}",
    eprint = "hep-th/9401070",
    archivePrefix = "arXiv",
    reportNumber = "SU-ITP-94-1",
    doi = "10.1103/PhysRevD.50.2700",
    journal = "Phys. Rev. D",
    volume = "50",
    pages = "2700--2711",
    year = "1994"
}

@article{Strominger:1996sh,
    author = "Strominger, Andrew and Vafa, Cumrun",
    title = "{Microscopic origin of the Bekenstein-Hawking entropy}",
    eprint = "hep-th/9601029",
    archivePrefix = "arXiv",
    reportNumber = "HUTP-96-A002, RU-96-01",
    doi = "10.1016/0370-2693(96)00345-0",
    journal = "Phys. Lett. B",
    volume = "379",
    pages = "99--104",
    year = "1996"
}

@article{Leichenauer:2018obf,
    author = "Leichenauer, Stefan and Levine, Adam and Shahbazi-Moghaddam, Arvin",
    title = "{Energy density from second shape variations of the von Neumann entropy}",
    eprint = "1802.02584",
    archivePrefix = "arXiv",
    primaryClass = "hep-th",
    doi = "10.1103/PhysRevD.98.086013",
    journal = "Phys. Rev. D",
    volume = "98",
    number = "8",
    pages = "086013",
    year = "2018"
}

@article{Bousso:2025xyc,
    author = "Bousso, Raphael",
    title = "{Robust Singularity Theorem}",
    eprint = "2501.17910",
    archivePrefix = "arXiv",
    primaryClass = "hep-th",
    doi = "10.1103/6f9b-3jmx",
    journal = "Phys. Rev. Lett.",
    volume = "135",
    number = "1",
    pages = "011501",
    year = "2025"
}

@article{Wall:2009wi,
    author = "Wall, Aron C.",
    title = "{Proving the Achronal Averaged Null Energy Condition from the Generalized Second Law}",
    eprint = "0910.5751",
    archivePrefix = "arXiv",
    primaryClass = "gr-qc",
    doi = "10.1103/PhysRevD.81.024038",
    journal = "Phys. Rev. D",
    volume = "81",
    pages = "024038",
    year = "2010"
}

@article{Wall:2010jtc,
    author = "Wall, Aron C.",
    title = "{The Generalized Second Law implies a Quantum Singularity Theorem}",
    eprint = "1010.5513",
    archivePrefix = "arXiv",
    primaryClass = "gr-qc",
    doi = "10.1088/0264-9381/30/19/199501",
    journal = "Class. Quant. Grav.",
    volume = "30",
    pages = "165003",
    year = "2013",
    note = "[Erratum: Class.Quant.Grav. 30, 199501 (2013)]"
}

@article{Faulkner:2017tkh,
    author = "Faulkner, Thomas and Haehl, Felix M. and Hijano, Eliot and Parrikar, Onkar and Rabideau, Charles and Van Raamsdonk, Mark",
    title = "{Nonlinear Gravity from Entanglement in Conformal Field Theories}",
    eprint = "1705.03026",
    archivePrefix = "arXiv",
    primaryClass = "hep-th",
    doi = "10.1007/JHEP08(2017)057",
    journal = "JHEP",
    volume = "08",
    pages = "057",
    year = "2017"
}

@article{Bousso:2015mna,
    author = "Bousso, Raphael and Fisher, Zachary and Leichenauer, Stefan and Wall, Aron C.",
    title = "{Quantum focusing conjecture}",
    eprint = "1506.02669",
    archivePrefix = "arXiv",
    primaryClass = "hep-th",
    doi = "10.1103/PhysRevD.93.064044",
    journal = "Phys. Rev. D",
    volume = "93",
    number = "6",
    pages = "064044",
    year = "2016"
}

@article{Faulkner:2016mzt,
    author = "Faulkner, Thomas and Leigh, Robert G. and Parrikar, Onkar and Wang, Huajia",
    title = "{Modular Hamiltonians for Deformed Half-Spaces and the Averaged Null Energy Condition}",
    eprint = "1605.08072",
    archivePrefix = "arXiv",
    primaryClass = "hep-th",
    doi = "10.1007/JHEP09(2016)038",
    journal = "JHEP",
    volume = "09",
    pages = "038",
    year = "2016"
}

@article{Bisognano:1976za,
    author = "Bisognano, J. J and Wichmann, E. H.",
    title = "{On the Duality Condition for Quantum Fields}",
    doi = "10.1063/1.522898",
    journal = "J. Math. Phys.",
    volume = "17",
    pages = "303--321",
    year = "1976"
}

@article{Casini:2011kv,
    author = "Casini, Horacio and Huerta, Marina and Myers, Robert C.",
    title = "{Towards a derivation of holographic entanglement entropy}",
    eprint = "1102.0440",
    archivePrefix = "arXiv",
    primaryClass = "hep-th",
    doi = "10.1007/JHEP05(2011)036",
    journal = "JHEP",
    volume = "05",
    pages = "036",
    year = "2011"
}

@article{Fliss:2024dxe,
    author = "Fliss, Jackson R. and Freivogel, Ben and Kontou, Eleni-Alexandra and Santos, Diego Pardo",
    title = "{How negative can null energy be in large N CFTs?}",
    eprint = "2412.10618",
    archivePrefix = "arXiv",
    primaryClass = "hep-th",
    month = "12",
    year = "2024"
}

@article{Ford:1999qv,
    author = "Ford, L. H. and Roman, Thomas A.",
    title = "{The Quantum interest conjecture}",
    eprint = "gr-qc/9901074",
    archivePrefix = "arXiv",
    doi = "10.1103/PhysRevD.60.104018",
    journal = "Phys. Rev. D",
    volume = "60",
    pages = "104018",
    year = "1999"
}

@article{Fewster:2007ec,
    author = "Fewster, Christopher J. and Osterbrink, Lutz W.",
    title = "{Quantum Energy Inequalities for the Non-Minimally Coupled Scalar Field}",
    eprint = "0708.2450",
    archivePrefix = "arXiv",
    primaryClass = "gr-qc",
    doi = "10.1088/1751-8113/41/2/025402",
    journal = "J. Phys. A",
    volume = "41",
    pages = "025402",
    year = "2008"
}

@article{Fewster:2002ne,
    author = "Fewster, Christopher J. and Roman, Thomas A.",
    title = "{Null energy conditions in quantum field theory}",
    eprint = "gr-qc/0209036",
    archivePrefix = "arXiv",
    doi = "10.1103/PhysRevD.67.044003",
    journal = "Phys. Rev. D",
    volume = "67",
    pages = "044003",
    year = "2003",
    note = "[Erratum: Phys.Rev.D 80, 069903 (2009)]"
}

@article{Wall:2011hj,
    author = "Wall, Aron C.",
    title = "{A proof of the generalized second law for rapidly changing fields and arbitrary horizon slices}",
    eprint = "1105.3445",
    archivePrefix = "arXiv",
    primaryClass = "gr-qc",
    doi = "10.1103/PhysRevD.85.104049",
    journal = "Phys. Rev. D",
    volume = "85",
    pages = "104049",
    year = "2012",
    note = "[Erratum: Phys.Rev.D 87, 069904 (2013)]"
}

@article{Burkardt:1995ct,
    author = "Burkardt, Matthias",
    title = "{Light front quantization}",
    eprint = "hep-ph/9505259",
    archivePrefix = "arXiv",
    doi = "10.1007/0-306-47067-5_1",
    journal = "Adv. Nucl. Phys.",
    volume = "23",
    pages = "1--74",
    year = "1996"
}

@article{Brown:2018hym,
    author = "Brown, Peter J. and Fewster, Christopher J. and Kontou, Eleni-Alexandra",
    title = "{A singularity theorem for Einstein{\textendash}Klein{\textendash}Gordon theory}",
    eprint = "1803.11094",
    archivePrefix = "arXiv",
    primaryClass = "gr-qc",
    doi = "10.1007/s10714-018-2446-5",
    journal = "Gen. Rel. Grav.",
    volume = "50",
    number = "10",
    pages = "121",
    year = "2018"
}

@article{Fewster:2019bjg,
    author = "Fewster, Christopher J. and Kontou, Eleni-Alexandra",
    title = "{A new derivation of singularity theorems with weakened energy hypotheses}",
    eprint = "1907.13604",
    archivePrefix = "arXiv",
    primaryClass = "gr-qc",
    doi = "10.1088/1361-6382/ab685b",
    journal = "Class. Quant. Grav.",
    volume = "37",
    number = "6",
    pages = "065010",
    year = "2020"
}

@article{Fewster:2021mmz,
    author = "Fewster, Christopher J. and Kontou, Eleni-Alexandra",
    title = "{A semiclassical singularity theorem}",
    eprint = "2108.12668",
    archivePrefix = "arXiv",
    primaryClass = "gr-qc",
    doi = "10.1088/1361-6382/ac566b",
    journal = "Class. Quant. Grav.",
    volume = "39",
    number = "7",
    pages = "075028",
    year = "2022"
}

@book{ONeill:1983semi,
  title={Semi-Riemannian geometry with applications to relativity},
  author={O'neill, Barrett},
  volume={103},
  year={1983},
  publisher={Academic press}
}

@article{Fliss:2021gdz,
    author = "Fliss, Jackson R. and Freivogel, Ben",
    title = "{Semi-local Bounds on Null Energy in QFT}",
    eprint = "2108.06068",
    archivePrefix = "arXiv",
    primaryClass = "hep-th",
    doi = "10.21468/SciPostPhys.12.3.084",
    journal = "SciPost Phys.",
    volume = "12",
    number = "3",
    pages = "084",
    year = "2022"
}

@article{Hartman:2023ccw,
    author = "Hartman, Thomas and Mathys, Gr{\'e}goire",
    title = "{Null energy constraints on two-dimensional RG flows}",
    eprint = "2310.15217",
    archivePrefix = "arXiv",
    primaryClass = "hep-th",
    doi = "10.1007/JHEP01(2024)102",
    journal = "JHEP",
    volume = "01",
    pages = "102",
    year = "2024"
}

@article{Hartman:2023qdn,
    author = "Hartman, Thomas and Mathys, Gr{\'e}goire",
    title = "{Averaged null energy and the renormalization group}",
    eprint = "2309.14409",
    archivePrefix = "arXiv",
    primaryClass = "hep-th",
    doi = "10.1007/JHEP12(2023)139",
    journal = "JHEP",
    volume = "12",
    pages = "139",
    year = "2023"
}

@article{Flanagan:1997gn,
    author = "Flanagan, Eanna E.",
    title = "{Quantum inequalities in two-dimensional Minkowski space-time}",
    eprint = "gr-qc/9706006",
    archivePrefix = "arXiv",
    doi = "10.1103/PhysRevD.56.4922",
    journal = "Phys. Rev. D",
    volume = "56",
    pages = "4922--4926",
    year = "1997"
}

@article{Fewster:1998pu,
    author = "Fewster, Christopher J. and Eveson, S. P.",
    title = "{Bounds on negative energy densities in flat space-time}",
    eprint = "gr-qc/9805024",
    archivePrefix = "arXiv",
    doi = "10.1103/PhysRevD.58.084010",
    journal = "Phys. Rev. D",
    volume = "58",
    pages = "084010",
    year = "1998"
}

@article{Epstein:1965zza,
    author = "Epstein, H. and Glaser, V. and Jaffe, A.",
    title = "{Nonpositivity of energy density in Quantized field theories}",
    doi = "10.1007/BF02749799",
    journal = "Nuovo Cim.",
    volume = "36",
    pages = "1016",
    year = "1965"
}

@mastersthesis{Pfenning:1998ua,
    author = "Pfenning, Michael John",
    title = "{Quantum inequality restrictions on negative energy densities in curved space-times}",
    eprint = "gr-qc/9805037",
    archivePrefix = "arXiv",
    type = "Other thesis",
    month = "4",
    year = "1998"
}

@book{Wald:1984rg,
    author = "Wald, Robert M.",
    title = "{General Relativity}",
    doi = "10.7208/chicago/9780226870373.001.0001",
    publisher = "Chicago Univ. Pr.",
    address = "Chicago, USA",
    year = "1984"
}

@article{Witten:2019qhl,
    author = "Witten, Edward",
    title = "{Light Rays, Singularities, and All That}",
    eprint = "1901.03928",
    archivePrefix = "arXiv",
    primaryClass = "hep-th",
    doi = "10.1103/RevModPhys.92.045004",
    journal = "Rev. Mod. Phys.",
    volume = "92",
    number = "4",
    pages = "045004",
    year = "2020"
}

@article{Friedman:1993ty,
    author = "Friedman, John L. and Schleich, Kristin and Witt, Donald M.",
    title = "{Topological censorship}",
    eprint = "gr-qc/9305017",
    archivePrefix = "arXiv",
    reportNumber = "PRINT-93-0448 (SANTA-BARBARA,ITP), NSF-ITP-93-80",
    doi = "10.1103/PhysRevLett.71.1486",
    journal = "Phys. Rev. Lett.",
    volume = "71",
    pages = "1486--1489",
    year = "1993",
    note = "[Erratum: Phys.Rev.Lett. 75, 1872 (1995)]"
}

@article{Leichenauer:2018tnq,
    author = "Leichenauer, Stefan and Levine, Adam",
    title = "{Upper and Lower Bounds on the Integrated Null Energy in Gravity}",
    eprint = "1808.09970",
    archivePrefix = "arXiv",
    primaryClass = "hep-th",
    doi = "10.1007/JHEP01(2019)133",
    journal = "JHEP",
    volume = "01",
    pages = "133",
    year = "2019"
}

@article{Hartman:2016lgu,
    author = "Hartman, Thomas and Kundu, Sandipan and Tajdini, Amirhossein",
    title = "{Averaged Null Energy Condition from Causality}",
    eprint = "1610.05308",
    archivePrefix = "arXiv",
    primaryClass = "hep-th",
    doi = "10.1007/JHEP07(2017)066",
    journal = "JHEP",
    volume = "07",
    pages = "066",
    year = "2017"
}

@article{Fliss:2023rzi,
    author = "Fliss, Jackson R. and Freivogel, Ben and Kontou, Eleni-Alexandra and Santos, Diego Pardo",
    title = "{Non-minimal coupling, negative null energy, and effective field theory}",
    eprint = "2309.10848",
    archivePrefix = "arXiv",
    primaryClass = "hep-th",
    month = "9",
    year = "2023"
}

@article{Fewster:2004nj,
	archiveprefix = {arXiv},
	author = {Fewster, Christopher J. and Hollands, Stefan},
	doi = {10.1142/S0129055X05002406},
	eprint = {math-ph/0412028},
	journal = {Rev. Math. Phys.},
	pages = {577},
	title = {{Quantum energy inequalities in two-dimensional conformal field theory}},
	volume = {17},
	year = {2005}}

@article{Balakrishnan:2019gxl,
    author = "Balakrishnan, Srivatsan and Chandrasekaran, Venkatesa and Faulkner, Thomas and Levine, Adam and Shahbazi-Moghaddam, Arvin",
    title = "{Entropy variations and light ray operators from replica defects}",
    eprint = "1906.08274",
    archivePrefix = "arXiv",
    primaryClass = "hep-th",
    doi = "10.1007/JHEP09(2022)217",
    journal = "JHEP",
    volume = "09",
    pages = "217",
    year = "2022"
}

@article{Bousso:2015wca,
	archiveprefix = {arXiv},
	author = {Bousso, Raphael and Fisher, Zachary and Koeller, Jason and Leichenauer, Stefan and Wall, Aron C.},
	doi = {10.1103/PhysRevD.93.024017},
	eprint = {1509.02542},
	journal = {Phys. Rev. D},
	number = {2},
	pages = {024017},
	primaryclass = {hep-th},
	title = {{Proof of the Quantum Null Energy Condition}},
	volume = {93},
	year = {2016}}

@article{Engelhardt:2026qqs,
    author = "Engelhardt, Netta and Nagar, Ivri",
    title = "{A Quantum Singularity Theorem for the Evaporating Black Hole}",
    eprint = "2605.05326",
    archivePrefix = "arXiv",
    primaryClass = "hep-th",
    reportNumber = "MIT-CTP/6034",
    month = "5",
    year = "2026"
}

@article{Fliss:2021phs,
	archiveprefix = {arXiv},
	author = {Fliss, Jackson R. and Freivogel, Ben and Kontou, Eleni-Alexandra},
	doi = {10.21468/SciPostPhys.14.2.024},
	eprint = {2111.05772},
	journal = {SciPost Phys.},
	number = {2},
	pages = {024},
	primaryclass = {hep-th},
	title = {{The double smeared null energy condition}},
	volume = {14},
	year = {2023}}

@article{Freivogel:2018gxj,
	archiveprefix = {arXiv},
	author = {Freivogel, Ben and Krommydas, Dimitrios},
	doi = {10.1007/JHEP12(2018)067},
	eprint = {1807.03808},
	journal = {JHEP},
	pages = {067},
	primaryclass = {hep-th},
	title = {{The Smeared Null Energy Condition}},
	volume = {12},
	year = {2018}}

@article{Freivogel:2020hiz,
	archiveprefix = {arXiv},
	author = {Freivogel, Ben and Kontou, Eleni-Alexandra and Krommydas, Dimitrios},
	doi = {10.21468/SciPostPhys.13.1.001},
	eprint = {2012.11569},
	journal = {SciPost Phys.},
	number = {1},
	pages = {001},
	primaryclass = {gr-qc},
	title = {{The Return of the Singularities: Applications of the Smeared Null Energy Condition}},
	volume = {13},
	year = {2022}}

@article{Kontou:2015yha,
	archiveprefix = {arXiv},
	author = {Kontou, Eleni-Alexandra and Olum, Ken D.},
	doi = {10.1103/PhysRevD.92.124009},
	eprint = {1507.00297},
	journal = {Phys. Rev. D},
	pages = {124009},
	primaryclass = {gr-qc},
	title = {{Proof of the averaged null energy condition in a classical curved spacetime using a null-projected quantum inequality}},
	volume = {92},
	year = {2015}}

@article{Kontou:2020bta,
	archiveprefix = {arXiv},
	author = {Kontou, Eleni-Alexandra and Sanders, Ko},
	doi = {10.1088/1361-6382/ab8fcf},
	eprint = {2003.01815},
	journal = {Class. Quant. Grav.},
	number = {19},
	pages = {193001},
	primaryclass = {gr-qc},
	title = {{Energy conditions in general relativity and quantum field theory}},
	volume = {37},
	year = {2020}}

@article{Balakrishnan:2017bjg,
	archiveprefix = {arXiv},
	author = {Balakrishnan, Srivatsan and Faulkner, Thomas and Khandker, Zuhair U. and Wang, Huajia},
	doi = {10.1007/JHEP09(2019)020},
	eprint = {1706.09432},
	journal = {JHEP},
	pages = {020},
	primaryclass = {hep-th},
	title = {{A General Proof of the Quantum Null Energy Condition}},
	volume = {09},
	year = {2019}}

@Article{2008_Hofman,
  author        = {Hofman, Diego M. and Maldacena, Juan},
  journal       = {JHEP},
  title         = {{Conformal collider physics: Energy and charge correlations}},
  year          = {2008},
  pages         = {012},
  volume        = {05},
  archiveprefix = {arXiv},
  doi           = {10.1088/1126-6708/2008/05/012},
  eprint        = {0803.1467},
  primaryclass  = {hep-th},
}

@article{Penrose:1964wq,
	author         = "Penrose, Roger",
	title          = "{Gravitational collapse and space-time singularities}",
	journal        = "Phys. Rev. Lett.",
	volume         = "14",
	year           = "1965",
	pages          = "57-59",
	doi            = "10.1103/PhysRevLett.14.57",
	SLACcitation   = "%%CITATION = PRLTA,14,57;%%"
}

@article{Hawking:1966sx,
	author         = "Hawking, Stephen W.",
	title          = "{The Occurrence of singularities in cosmology}",
	journal        = "Proc. Roy. Soc. Lond.",
	volume         = "A294",
	year           = "1966",
	pages          = "511-521",
	doi            = "10.1098/rspa.1966.0221",
	SLACcitation   = "%%CITATION = PRSLA,A294,511;%%"
}

\end{document}